\documentclass[aps,prx,twocolumn,showpacs,superscriptaddress,preprintnumbers,10pt]{revtex4-2} 
\usepackage{float}
\usepackage{dcolumn}   % needed for some tables
\usepackage{bm}        % for math
\usepackage{amssymb}   % for math
\usepackage{amsmath}
\usepackage{comment}
\usepackage{xcolor}
\usepackage{slashed}
\usepackage{braket}
\usepackage{color}
\usepackage{amsthm}
\usepackage[caption=false]{subfig}
\usepackage{graphicx}  % needed for figures
\usepackage[export]{adjustbox}
\usepackage{multirow}
\usepackage{rotating,booktabs}
\usepackage[verbose]{placeins}
\usepackage{mathrsfs}
\usepackage{stackengine}
\usepackage{soul}
\usepackage{algorithmic}
\usepackage[colorlinks=true,citecolor=blue,linkcolor=red]{hyperref}

\newcommand{\CC}{{\mathbb{C}}}

\newcommand{\poly}{{\mathrm{poly}}}

\usepackage[normalem]{ulem}

\newcommand{\lf}[0]{\left( }
\newcommand{\ri}[0]{ \right)}

\newcommand{\hwt}[1]{\left|\left.#1 \right\rangle\!\right\rangle_{h}}
\newcommand{\ketbra}[2]{\ket{#1}\!\bra{#2}}

\newtheorem{prop}{Proposition}
\newtheorem{thm}{Theorem}
\newtheorem*{thm*}{Theorem}

\newtheorem{lemma}{Lemma}
\newtheorem*{lemma*}{Lemma}

\definecolor{cadmiumgreen}{rgb}{0.0, 0.42, 0.24}

\newcounter{algocount}
\makeatletter
\def\blfootnote{\xdef\@thefnmark{}\@footnotetext}
\makeatother
\def \addCQuIC {Center for Quantum Information and Control, University of New Mexico, Albuquerque, 87131, NM, USA}
\def \addPandAUNM {Department of Physics and Astronomy, University of New Mexico, Albuquerque, NM, 87106, USA}

\begin{document}

\title{Efficient simulation of noisy IQP circuits with amplitude-damping noise}
\author{Shravan Shravan}
\email{shravan@unm.edu}
\affiliation{\addCQuIC} \affiliation{\addPandAUNM}
\author{Mohsin Raza}
\email{mraza98@unm.edu}
\affiliation{\addCQuIC} \affiliation{\addPandAUNM}
\author{Ariel Shlosberg}
\affiliation{\addCQuIC} \affiliation{\addPandAUNM}

\begin{abstract}
Efficient classical simulation of noisy intermediate-scale quantum (NISQ) circuits has been a topic of intense study over the past few years. The majority of results on efficient simulation assume that the circuits undergo some variant of \textit{unital noise} or involve sufficient \textit{randomness}. However, there are limited results for circuits undergoing non-unital noise in the absence of randomness. In this work, we present a polynomial-time classical algorithm to sample from the output distributions of amplitude-damped instantaneous quantum polynomial (IQP) circuits. Our algorithm works for circuits generated by arbitrary $\ell$-local diagonal gates with depth $d = \Omega(\log(n))$, undergoing constant amplitude-damping noise. 
\end{abstract}
\maketitle

\noindent
\emph{Introduction.} Quantum computers are expected to provide substantial speedups over their classical counterparts for certain computational tasks, in what is known as \textit{quantum advantage}~\cite{Preskill:Quant_Adv,Shor_1994,HHL_2009,lloyd_universal_1996}. However, noise remains a significant barrier to realizing this advantage in practice. Although remarkable progress has been made in recent years towards the implementation of fault-tolerant quantum processors ~\cite{Exp_IQP,QEC_Google,QEC_China}, large-scale quantum algorithms require resources far beyond current technological capabilities~\cite{PreskillNISQ, eisert2025mindgapsfraughtroad,gidney2025factor}. This challenge motivates the study of using near-term quantum circuits to demonstrate computational advantage over classical computers without requiring full fault tolerance. 

One of the most well-known classes of tasks proposed for this purpose is the set of quantum sampling problems \cite{bouland2019complexity,ArkhipovAaronson13,bremner_classical_2010}. These involve sampling from the output distributions of quantum circuits that are drawn from  specific ensembles. Sampling from these distributions is widely considered to be classically intractable, supported by complexity-theoretic evidence establishing asymptotic hardness~\cite{bouland2019complexity,ArkhipovAaronson13,bremner_classical_2010,bremner_average_IQP_2016,denzler2026simulation}. Numerical studies have also demonstrated that classical simulation remains difficult even for instances of modest system size~\cite{Numec_1,hangleiter_RCS_review_2023}. Several experimental realizations of sampling protocols have been reported in the literature that provide evidence for quantum advantage~\cite{Exp1,Exp2,Exp_IQP,morvan_phase_2024}. However, in conjunction with these experiments, classical algorithms have been developed to try to efficiently simulate the same types of circuits in hardware-limited regimes~\cite{Numec_1,maslov_fast_2024,small_classical_1,small_classical_2,small_classical_3}. 

Quantum noise is modeled by channels that are typically classified into two categories: unital (i.e., the maximally mixed state is a fixed point) and non-unital~\cite{wilde2013quantum}. Efficient estimation of observable expectation values for noisy circuits has been established for both unital \cite{LOWESA, schuster2025polynomial,cirstoiu2024Fourier, Pauli_back_propogation_VQA,gonzalez-garcia_pauli_2024, Poly_back_prop_average_input} and non-unital \cite{angirasani,Pauli_back_prop_non_unital_VQA} noise. With sufficient randomness, it is even possible to estimate the expectation values of noiseless quantum circuits \cite{angrisani2024classically}. In contrast, classically sampling from the output distributions of noisy quantum circuits is less well understood and has been established only for Pauli noise~\cite{Gao_Efficient_simulationy, cirstoiu2024Fourier, BMS17, rajakumar_polynomial-time_2025}.
 
There are two reasons for the lack of classical algorithms for sampling under non-unital noise. First, non-unital channels are incompatible with the anti-concentration property, which is often used to prove classical spoofability~\cite{fefferman2024effect}. Second, there exist instances of circuits undergoing non-unital noise that simulate noiseless quantum circuits, implying that arbitrary circuits under non-unital noise should not be classically simulable (unless BQP=BPP)~\cite{ben-or_quantum_2013}. This leaves open the question of whether specific families of circuits, which are conjectured to be classically hard in the absence of noise, can become efficiently simulable under physically-relevant non-unital noise. Recent progress has been made for geometrically local circuits, where it was shown that classical simulability can be established in the absence of anti-concentration, provided that the conditional mutual information (CMI) decays sufficiently~\cite{BB_sampling_unital,BB_sampling_unital_1,BB_sampling_non_unital}. However, rigorous guaranties for the decay of CMI have been established only for depolarizing noise~\cite{BB_sampling_unital,BB_sampling_unital_1}. For non-unital noise, evidence remains purely numerical and is limited to specific cases, such as 1D Haar-random and 2D Clifford-random architectures~\cite{BB_sampling_non_unital}.

In this work, we demonstrate a polynomial-time classical algorithm for amplitude-damped IQP circuits. Assuming constant amplitude-damping noise after each unitary layer, we show that we can approximately sample from the output distributions of a broad class of IQP circuits beyond $O(\log(n))$ depth. IQP circuits were first introduced in Ref.~\cite{shepherd_temporally_2009} as circuits with no temporal ordering. Classical hardness of sampling from IQP circuits was first established in Ref.~\cite{bremner_classical_2010}, where it was shown that exact sampling from IQP circuits would imply the collapse of the polynomial hierarchy. Hardness of approximate sampling was established in Ref.~\cite{bremner_average_IQP_2016} by connecting the task of sampling to approximating the partition functions of Ising models. Noisy IQP circuits were first studied in Ref.~\cite{BMS17}, assuming measurement bit-flip error. Under anti-concentration assumptions, they proved that the typical output distributions of such circuits are close to uniform and hence classically simulable. Finally, in Ref.~\cite{rajakumar_polynomial-time_2025}, it was shown that arbitrary IQP circuits with a depth greater than a critical $O(1)$ threshold, undergoing local Pauli noise, can be classically simulated. 

\emph{Setup.} IQP circuits involve applying gates that are diagonal in the computational basis on qubits initialized in the $\ket{+}^{\otimes n}$ state. At the end of the circuit, the qubits are measured in the Hadamard basis (Pauli-X basis). 
\begin{figure}
    \centering
\includegraphics[width=0.95\linewidth]{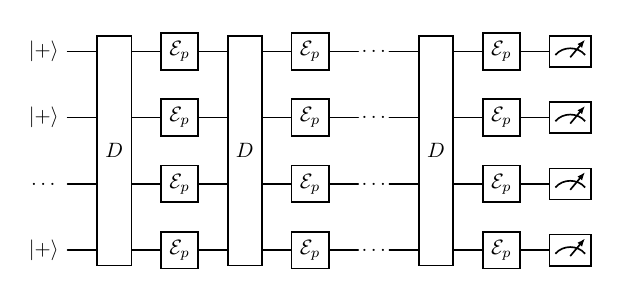}
    \caption{A prototypical noisy IQP circuit. At the beginning of the circuit, the state is given by $\ket{+}^{\otimes n}.$ The unitary gate layers $D$ are composed of $\ell$-local gates from the gate set $\mathcal{G}_l$. After each unitary layer, we have identical, local amplitude-damping channels of strength $p$ acting on each qubit. The circuit ends with a Pauli-X basis measurement on every qubit. }
    \label{fig:IQP_circ}     
\end{figure}
We consider sampling from the output distribution of an IQP circuit $\mathcal{C}$ interspersed with local amplitude-damping noise of constant strength $p$ (see Fig.~\ref{fig:IQP_circ}). The circuit does not include the initial and final layer of Hadamards; rather, they are considered part of state preparation and measurement, respectively. The diagonal unitary gates can be constructed using sets of local gates such as $\{T,CS\}, \{e^{itZ}, e^{i\theta ZZ}\}$, or $\{Z,CZ,CCZ\}$~\cite{bremner_average_IQP_2016}. In this work, we will consider arbitrary single-qubit $Z$-rotations and $\ell$-local controlled-phase gates, $\mathcal{G}_l = \{e^{itZ}, C(\theta) =  \text{diag}(1,1,1,e^{i \theta}), C^{(2)}(\theta), \cdots, C^{(\ell)} (\theta)\} $ where $C^{(\ell)} (\theta)$ is a controlled-Z gate, controlled on $\ell-1$ qubits and acting on the $\ell$-th qubit. This is equivalent to the more standard set $\{e^{i\theta_{1}Z}, e^{i\theta_{2} ZZ}, \dots, e^{i\theta_{\ell} Z \cdots Z}\}$ (with a maximum support on $\ell$ qubits) as any circuit implemented by one of these gate sets can be exactly implemented by the other at the cost of a constant multiplicative overhead to the circuit depth. We do not place any restrictions on the geometric architecture of the system. After each layer of gates, we assume that a local amplitude-damping channel $\mathcal{E}_p$ of uniform strength $p$ acts on all qubits,
\begin{align}
   &\mathcal{E}_p (\rho) = \sum_{i = 0}^{1} K_i \rho K_i^\dagger, \\
   &K_0 = \ketbra{0}{0} + \sqrt{1 - p} \ketbra{1}{1},\ K_1 =\sqrt{p} \ketbra{0}{1}, \nonumber
\end{align}
where $0 < p \leq 1$. Given a circuit $\mathcal{C}$, we will use the notation $P_{\mathcal{C}}$ to denote the Born-rule probability distribution associated with a Hadamard-basis measurement of $\rho = \mathcal{C}(\ketbra{+}{+}^{\otimes n})$. 

At a high level, our algorithm relies on three key ingredients: (a) the amplitude-damping channel has a unique fixed point, pushing the states in Hilbert space towards a low Hamming-weight subspace, (b) the diagonal gate set preserves this subspace, and (c) circuits consisting of arbitrary $\ell$-local diagonal gates admit a frame (that we define below) that allows for efficient simulation of this subspace. Due to amplitude damping, only a constant number of frame elements have to be tracked in order to obtain a distribution that is close to the true output distribution of the noisy IQP circuit.

\emph{Results.}
We will now present the main result of this Letter and give a detailed overview of our proof. The rigorous derivations are presented in the Supplementary Material (SM)~\cite{SM}. 
\begin{thm}
    \label{thm:main_theorem}
    Suppose $\mathcal{C}$ is a noisy IQP circuit of depth $d$ acting on $n$ qubits composed of $\ell$-local gates from $\mathcal{G}_\ell$, where each layer is interspersed with identical amplitude-damping channels, $\mathcal{E}^{\otimes n}_p$, of strength $p = \Omega(1)$. There exists a threshold $d_T = O\left(\log(n)/\log(1-p)^{-1} \right)$, such that when $d \geq d_T$, there exists a classical algorithm that samples from a distribution $Q_{\mathcal{C}}$, where $\|Q_{\mathcal{C}} - P_{\mathcal{C}}\|_{TVD} \leq \epsilon$, with a worst-case runtime of $T \leq O\left(n^3d\, \poly(1/\epsilon)\right)$.          
\end{thm}

For the sake of clarity, in the main text, we prove Theorem~\ref{thm:main_theorem} for the restricted case of the $2$-local gate set $\mathcal{G}_2$. Noting that most of the aspects of the proof remain the same in the $\ell$-local case, we provide the extension of the proof for $\mathcal{G}_\ell$ in Sec.~\ref{ap:Theorem_main_extension} of the SM. The proof is structured as follows: First, we introduce our notation and truncation scheme and prove Lemma \ref{lem:Trunc_error}, which bounds the Hilbert-Schmidt error due to truncation. Second, using Theorem~\ref{thm:Bound_on_TD}, we show that a small Hilbert-Schmidt error leads to only a small error in trace distance. Third, we define an overcomplete basis (frame) and show that this allows us to track the non-truncated coefficients exactly. Taken together, this allows us to efficiently sample from the output distribution of sufficiently-deep, noisy IQP circuits.

In the following, we use a permuted vectorized basis for operators on the $n$-qubit Hilbert space that we call the \textit{Hamming weight basis} (HW basis). Let $\left|\left.\rho\right\rangle\!\right\rangle
$ denote the computational-basis vectorized representation of the density matrix $\rho$ \cite{gilchrist_vectorization_2011}. The HW basis state, $\hwt{\rho}$, is a permutation of $\left|\left.\rho\right\rangle\!\right\rangle$, such that the operators are lexicographically organized in order of increasing Hamming weight. (See Sec.~\ref{ap:HWT_basis_example} of the SM for a worked-out example.) We use the notation $i \in[1,4^n]$ to denote the index of the operator in the HW basis. The symbol $[i]$ denotes the Hamming weight of the operator indexed by $i$.

\emph{Sparse HW representation.} Under amplitude-damping noise, the coefficient of an operator indexed by $i$ in the HW basis is suppressed by a factor of ${(1 - p)}^{[i]}$. This allows us to truncate operators with high Hamming weights, while incurring only a small error. We define $\hwt{\sigma}$ as the truncated approximation to $\hwt{\rho}$, with $\mathcal{T}$ denoting the set of truncated indices. The Hilbert-Schmidt error, $\varepsilon$, is defined as 
\begin{equation}
\label{eq:HS_error}
    \varepsilon^2 := \|\rho - \sigma\|_{HS}^2  = \left\langle\!\left\langle \rho - \sigma| \rho - \sigma \right\rangle\!\right\rangle_{h} = \sum_{i \in \mathcal{T}} \left|\alpha_i\right|^2,
\end{equation}  
 where $\alpha_i$ is the coefficient of the operator indexed by $i$ in $\hwt{\rho}$. 
 Amplitude damping also exhibits \textit{refeeding}: coefficients of operators containing factors of $\ketbra{0}{0}$ are increased by contributions proportional to the coefficients of the same operators in which those factors are replaced by $\ketbra{1}{1}$. Refeeding is an essential feature of amplitude damping and ensures that the map is trace preserving. The complete action of the amplitude-damping channel on HW basis operators, including damping and refeeding, is derived in Sec.~\ref{ap:AD_noise_structure} of the SM. Consequently, the evolution of the coefficients depends not only on the Hamming weight but also on the number of $\ketbra{0}{0}$ factors present in the operators. In Proposition \ref{prop:Truncation_error}, proved in Sec.~\ref{ap:bound_based_00_terms} of the SM, we show this rigorously. Using Proposition \ref{prop:Truncation_error} allows us to prove the following lemma: 
\begin{lemma}
    \label{lem:Trunc_error}
    Using the same notation as in Theorem \ref{thm:main_theorem}, there exists a depth threshold $d^*_T = \log(n)/\log(1-p)^{-1}$, such that for $d > d^*_T$, the truncation error in Hilbert-Schmidt distance, incurred by truncating operators with Hamming weight greater than $k$, is upperbounded by 
    \begin{equation}
        \varepsilon^2   \leq \frac{\lf 2 - (1 - p)^d \ri^{2n - k - 1}}{4^n} e^{2n H\lf\frac{k+1}{2n} \ri} (1 - p)^{d(k+1)},\label{eq:hilbert-schmidt-error}
    \end{equation}
    where $H(x) = -x\ln x - (1 - x)\ln(1 - x)$ is the binary entropy.
\end{lemma}
 The lemma is proved in Sec.~\ref{ap:Tail_bound} of the SM by relating the truncation error to the tail of a binomial distribution, which is then upperbounded through Chernoff bounds~\cite{chernoff1952measure}. The use of Chernoff bounds imposes the threshold that appears in Lemma~\ref{lem:Trunc_error}. Moreover, the bound in Lemma~\ref{lem:Trunc_error} is tight, as $\varepsilon^2 = 4^{-n}(1 - p)^{2nd}$ when $k+1 = 2n$. 

Next, we translate the bound on the Hilbert–Schmidt error into a bound on the trace-distance error, $\|\rho - \sigma\|_{TD}$, to eventually bound the total variation distance. Generalizing the result of Coles and Cerezo~\cite{coles_cerezo_bound_2019}, we prove Theorem~\ref{thm:Bound_on_TD}.
\begin{thm}
\label{thm:Bound_on_TD}
    Suppose $\rho$ is a density matrix and $\sigma$ is a non-zero Hermitian matrix. Furthermore, suppose $\operatorname{Tr}(\sigma) < 1$ and there exists an $\epsilon' \geq 0$ such that 
    \begin{subequations}
        \begin{align}
        \label{Eq:TD_bound_req_1}
        &|\operatorname{Tr}(\rho - \sigma) | \leq \epsilon',  \\
        \label{Eq:TD_bound_req_2}
        & \|\rho -\sigma \|_{\text{HS}} \leq \epsilon'.
    \end{align}
    \end{subequations}
    Then, 
    \begin{equation}
        \|\rho -\sigma \|_{\text{TD}} \leq \left(\sqrt{\operatorname{rank}(\sigma)} + 1 \right)\epsilon'.
    \end{equation}
\end{thm}
The proof is given in Sec.~\ref{ap:Td_upper_bound} of the SM. In Sec.~\ref{ap:Bounding_trace} of the SM, we show that the requirement Eq.~\eqref{Eq:TD_bound_req_1} is satisfied by the truncated state $\sigma$ with
\begin{equation}
    \label{eq:trace_bnd}
     |\operatorname{Tr}(\rho - \sigma) | \leq \frac{\lf 2 - (1 - p)^d \ri^{n - \frac{k + 1}{2}}}{2^n} e^{n H\lf\frac{k+1}{2n} \ri} (1 - p)^{\frac{d(k+1)}{2}}.
\end{equation}
Combining Eq.~\eqref{eq:trace_bnd}, Lemma \ref{lem:Trunc_error} (satisfying requirement \eqref{Eq:TD_bound_req_2}), and Theorem~\ref{thm:Bound_on_TD}, in Sec.~\ref{ap:k_scaling} of the SM, we bound $\|\rho - \sigma\|_{TD}$. Specifically, we show that for all $\delta > 0$, taking the maximum Hamming weight of $\sigma$ to be at least
\begin{equation}
    \label{eq:k_scaling}
    k  =   \frac{ \ln (2/\delta)}{ \lf \frac{\lambda}{2} -  2 \ri  \ln \lf n \ri  - \ln 4}-1 = O \lf \frac{\log (2/\delta)}{\log(n)}  \ri,
\end{equation}
ensures that $\|\rho - \sigma\|_{TD} < \delta$, as long as the circuit depth $d \geq d_T$.
Here, we have parameterized the circuit depth as $d(\lambda) = \lambda \log(n)/\log(1 - p)^{-1}$.
The threshold takes the form 
\begin{equation}
      d_T  =  \frac{2}{\ln(1-p)^{-1}} \lf 2 \ln(n) + \ln(2)  \ri = O\lf \frac{\log(n)}{\log(1-p)^{-1}} \ri.
\end{equation}
For a given $k$ that scales as in Eq.~\eqref{eq:k_scaling}, the total number of non-truncated strings in $\hwt{\sigma}$ is given by 
\begin{equation}
    \label{eq:num_strings}
    \sum_{m=0}^k \binom{2n}{m} = O(n^{k+1}) = O(\poly(\delta^{-1})).
\end{equation}
Hence, for a constant error, $\hwt{\sigma}$ is sparse in the HW basis and can be stored efficiently. However, it remains to be shown that the coefficients $\alpha_i$ of $\hwt{\sigma}$ can be efficiently computed from the circuit description.

\emph{Computing coefficients.} We now show that the coefficients $\alpha_i$  can be exactly computed in an efficient manner. Define the following operator, for all $a \in \CC$, 
\begin{equation}
   \label{eq:clubbed_identity}
    I(a) = \ketbra{0}{0} + a \ketbra{1}{1}.
\end{equation}
The set $\mathcal{I} := \{I(a), \sigma_+, \sigma_-| a \in \CC, |a| \leq 1\}$, where $\sigma_{\pm}$ are the usual raising and lowering operators, forms an operator frame on a single qubit. Consider the action of the elements of the noisy IQP circuit $\mathcal{C}$ on an $n$-qubit operator string $s \in \mathcal{I}^{\otimes n}$ with an associated coefficient $\beta$. Single-qubit diagonal unitaries apply a phase when acting on $\sigma_\pm$. Thus, single-qubit gates only multiply $\beta$ by a phase. The action of two-qubit controlled-phase gates on the elements of the frame is as follows,
\begin{subequations}
    \begin{align}
    \label{eq:identities_line_a}
    &C(\theta) I(a) \otimes  I(b) C^\dagger (\theta) = I(a) \otimes I(b), \\
    \label{eq:identities_line_b}
    &C(\theta) \sigma_{\pm} \otimes \sigma_{\pm} C^\dagger (\theta) =  e^{i (\theta/2) \lf n_+  - n_-\ri } \sigma_{\pm} \otimes \sigma_{\pm}, \\ 
\label{eq:identities_line_c}
    &C(\theta) \sigma_{\pm} \otimes I(a)  C^\dagger (\theta)  =  \sigma_\pm \otimes I\lf ae^{ \pm i \theta}\ri,
\end{align}
\end{subequations}
where in Eq.\eqref{eq:identities_line_b} $n_{+}$($n_-$) denotes the total number of $\sigma_+$($\sigma_-$). Similarly, amplitude damping acts on elements of $\mathcal{I}$ as 
\begin{subequations}
    \begin{align}
    \label{eq:identities_line_d}
    & \mathcal{E}_p \lf \sigma_{\pm} \ri = \sqrt{1 - p} \sigma_{\pm}, \\
    \label{eq:identities_line_e}
    &\mathcal{E}_p \lf I(a) \ri = (1 + a p) I\lf \frac{a(1 - p)}{1 + a p} \ri.
\end{align}
\end{subequations}
From Eqs.~\eqref{eq:identities_line_a}--\eqref{eq:identities_line_e}, it is evident that the action of $\mathcal{C}$ on an operator string $s \in \mathcal{I}^{\otimes n}$ only changes the coefficients  $\beta$ and arguments $\{a_i\}$. This amounts to a one-to-one mapping between strings in $\mathcal{I}^{\otimes n}$. In Sec.~\ref{Ap:Low_Hamming_Weight_Strings} of the SM, we make use of this observation and explicitly construct an algorithm that propagates any string $s  \in \mathcal{I}^{\otimes n}$ under $\mathcal{C}$, with a worst-case runtime of $O(n^3 d)$. 

We use this algorithm now to estimate $\alpha_{i}$. Consider the following representation of the initial state,
\begin{equation}
   \label{eq:Initial_state_frame_rep}
    \ketbra{+}{+}^{\otimes n} = \frac{1}{2^n} \lf I(1) + \sigma_+ + \sigma_-\ri^{\otimes n} =  \frac{1}{2^n}\sum_{m = 0}^n \sum_{s \in S_m} s,
\end{equation}
where $S_m$ is the set of all operator strings in the expansion with $m$ $\sigma_\pm$-terms. The size of $S_m$ is
\begin{equation}
 \label{eq:size_scaling}
   |S_m| = 2^m {\binom{n}{m}}.   
\end{equation} 
Therefore, Eq.~\eqref{eq:Initial_state_frame_rep} uses an exponential number of strings $s \in \mathcal{I}^{\otimes n}$ to represent the initial state. First, note that every operator in the HW basis is contained in a unique string in the above expansion. That is, if $o_i$ denotes a HW operator, then there exists only one operator string $s$ such that
\begin{equation}
    \operatorname{Tr}  \lf o_i ^\dagger s \ri \neq 0.
\end{equation}
Propagating this string $s$ and its associated coefficient $\beta$ through $\mathcal{C}$, allows us to exactly recover the coefficient $\alpha^{(r)}_i$ associated with $o_i$. Concretely, 
\begin{equation}
    \alpha^{(r)}_i = \beta^{(r)} \operatorname{Tr} \lf o_i^\dagger s^{(r)} \ri, 
\end{equation}
 where $ \beta^{(r)}(s^{(r)})$ is the coefficient (operator) after $r$ layers of evolution under $\mathcal{C}$. Second,  all elements of the set $S_m$ have exactly $m$ tensor factors of $\sigma_{\pm}$. Thus, for any HW operator $s \in S_m$, the Hamming weight satisfies $[i] \geq m$.
 To determine $\hwt{\sigma}$, we only need to propogate strings in $S_{m}$ for $m \leq k$. From Eq.~\eqref{eq:num_strings}, we only need to propagate $O(\poly(\delta^{-1}))$ individual strings. Thus, the worst-case runtime to determine $\sigma$ is $O(n^3d \,\poly(\delta^{-1}))$. 

\emph{Sampling.} Having shown that we can efficiently obtain $\sigma$, we now show how to generate samples from  the probability distribution $Q_{\mathcal{C}}$. Due to truncation, $\sigma$ is only guaranteed to be Hermitian but not necessarily positive. 
For any bit-string $x \in \{+,-\}^n$ in the Hadamard basis,
consider $q(x)$ and the map $f :\{+/-\}^n \to \{0/1\}^n$,  
\begin{equation}
\label{eq:Probability_dist}
    q(x) = \bra{x} \sigma \ket{x} = \sum_{i \in \overline{\mathcal{T}}} \alpha_i^{(d)} (-1)^{f(\vec{x})\cdot(\vec{a}_i + \vec{b}_i)},
\end{equation}
where $q(x)$ approximates $P_{\mathcal{C}}(x)$, and $a_i,b_i$ are computational strings such that the HW operator indexed by $i$ is $\ketbra{a_i}{b_i}$. Since $\sigma$ is not necessarily positive, $q(x)$ is only a quasiprobability distribution satisfying $0<\sum_x q(x) \leq 1$. From Eq.~\eqref{eq:Probability_dist}, it is evident that appropriate combinations of $\alpha_i^{(d)}$ are the Fourier coefficients of $q(x)$. Therefore, the number of distinct Fourier coefficients of $q(x)$ can be at most $|\overline{\mathcal{T}}| = O(\poly(\delta^{-1}))$, meaning it has a sparse Fourier spectrum. This implies that the marginals of $q(x)$ can be efficiently computed using Parseval's identity~\cite{BMS17}. Let $y \in \{0,1\}^k$ be an arbitrary $k$-bit string, then the marginal $S_y$ is given as in Ref.~\cite{BMS17},  
\begin{equation}
    S_y = \sum_{x:x_{1 \dots k } = y} q(x) = 2^{n-k} \sum_{s: s_{k+1,\dots,n} = 0^{n - k}} \tilde{q}(s)(-1)^{y\cdot s_1 \dots s_k}, 
\end{equation}
where $\tilde{q}(s)$ are the Fourier coefficients of $q(x)$. Since the Fourier spectrum is sparse, the marginals are exactly computable in $O(n \poly(\delta^{-1}))$ time. Having access to the marginals, the algorithm proposed in Lemma 10 of Ref.~\cite{BMS17} can be used to sample from the probability distribution $Q_{\mathcal{C}} = \text{Alg}(q(x))$.  

 Finally, we prove that by appropriately choosing $k$, we can ensure $\|P_{\mathcal{C}} - q \|_{TVD} \leq \delta$ for any $\delta > 0$. To do this, we first use the following inequality that connects the trace distance and the total variation distance~\cite{Nielsen_Chuang_2010},
\begin{equation}
\label{Eq:TVD_bound_by_Td}
    \|P_{\mathcal{C}} - q\|_{TVD} \leq \|\rho - \sigma\|_{TD}\leq \delta.
\end{equation}
 Eq.~\eqref{Eq:TVD_bound_by_Td} holds even if $\rho$ and $\sigma$ are only Hermitian but not necessarily positive. Then, with $\delta = \epsilon/(4 + \epsilon)$ and using Lemma 10 from Ref.~\cite{BMS17}, the probability distribution $Q_{\mathcal{C}} = \text{Alg}(q(x))$ satisfies 
\begin{equation}
     \|P_{\mathcal{C}}(x) - Q_{\mathcal{C}}(x)\|_{TVD} \leq \frac{4 \delta}{1 - \delta } = \epsilon,
\end{equation}
completing the proof. 

The extension of the proof of Theorem~\ref{thm:main_theorem} for the $\ell$-local gate set $\mathcal{G_\ell}$ is given in Sec. \ref{ap:Theorem_main_extension} of the SM. The only non-trivial aspect stems from the fact that unlike in the case of two-qubit gates, higher-order gates (e.g. CCZ) can map a single frame element to multiple elements. However, we prove that for an arbitrary $\ell$-local controlled-phase gate, this increased branching results in only a constant-factor overhead in the number of terms that must be tracked. Finally, we implement a numerical simulation of the algorithm proposed in this Letter and show that our results hold up to numerical scrutiny. These results can be found in Sec.~\ref{ap:numerical} of the SM.

\emph{Discussion.} In this work, we provide an efficient classical algorithm to sample from the output distributions of $\ell$-local noisy IQP circuits, subject to local amplitude-damping noise. Our algorithm is based on the observation that amplitude-damping noise pushes the state of the circuit closer to its fixed point. Therefore, we can  approximate the state with only a constant number of frame elements. Further, the gate-set $\mathcal{G}_l$ lends itself to efficiently tracking how these frame elements evolve. Together, this gives a recipe for constructing an efficient classical algorithm. More generally, this approach can be applied to other noise models with unique fixed points and associated restricted circuit families. For instance, Pauli-path truncation schemes for depolarizing noise (whose fixed point is the maximally mixed state) can be seen as a similar algorithm, but with the Pauli basis as its frame. Our results also have implications for the quantum refrigerator arguments presented in Ref.~\cite{ben-or_quantum_2013}. We have shown that IQP circuits are a class of circuits that cannot be `refrigerated'. This begs the question: What are other classes of circuits that show similar behavior? We conjecture that a circuit needs to have sufficiently high density of Hadamard gates to preserve quantum information.   

A natural extension of this work is to establish classical simulability for shallower circuits with depths below the $O(\log(n))$ threshold reported in this Letter. For local Pauli noise, a constant-depth threshold has recently been proven in Ref.~\cite{rajakumar_polynomial-time_2025} by observing that the noise effectively disentangles the qubits. Whether the threshold reported in this Letter is tight or not, remains an open question. We conjecture that it is not so. 

\emph{Acknowledgements.} We thank Akimasa Miyake, Gopikrishnan Muraleedharan, Ivan Deutsch, Mohammad Alhejji, and Niklas Mueller for discussions. This work is supported by a collaboration between the US DOE and the NSF. This material is based upon work supported by the U.S. Department of Energy, Office of Science, National Quantum Information Science Research Centers, Quantum Systems Accelerator (Award No. DE-SCL0000121). Additional support is acknowledged from the National Science Foundation Grant No. PHY-2116246.

\bibliographystyle{apsrev4-2}
\bibliography{ref.bib}

\appendix
\onecolumngrid
\newpage

\section{Hamming weight basis}
\label{ap:HWT_basis_example}
In this section, we provide an explicit example to illustrate what we refer to as the Hamming weight (HW) basis. We define this basis by ordering vectorized computational basis strings according to increasing Hamming weight. This representation is useful because it provides a transparent view of the coefficients that are truncated in our algorithm. Consider an arbitrary two-qubit density matrix represented in the computational basis as, 
\begin{equation}
    \rho = \sum_{a,b,c,d = 0}^1 \alpha_{ab|cd } \ketbra{ab}{cd}.
\end{equation}
In the vectorized representation, we have \cite{gilchrist_vectorization_2011}, 
\begin{equation}
    \left|\left.\rho\right\rangle\!\right\rangle = \sum_{a,b,c,d = 0}^1 \alpha_{ab|cd } \ket{abcd}.
\end{equation}
The Hamming weight representation is nothing but a permutation of the vectorized $\left|\left.\rho\right\rangle\!\right\rangle$ such that the operators are arranged in the order of increasing hamming weight. For strings with the same Hamming weight, we order them lexicographically.  Hence, as a column vector, we have, 
\begin{equation}
    \hspace{-0.5cm}\hwt{\rho} = \begin{bmatrix}
        \alpha_{00|00} ,| \textcolor{blue}{\alpha_{00|01}}, \textcolor{blue}{\alpha_{00|10}}, \textcolor{blue}{\alpha_{01|00}}, \textcolor{blue}{\alpha_{10|00}},| \textcolor{orange}{\alpha_{00|11}}, \textcolor{orange}{\alpha_{01|01}},\textcolor{orange}{\alpha_{01|10}}, \textcolor{orange}{\alpha_{10|01}}, \textcolor{orange}{\alpha_{10|10}} ,\textcolor{orange}{\alpha_{11|00}}, |\textcolor{teal}{\alpha_{01|11}}, \textcolor{teal}{\alpha_{10|11}}, \textcolor{teal}{\alpha_{11|01}}, \textcolor{teal}{\alpha_{11|10}}, |\textcolor{purple}{\alpha_{11|11}}.
    \end{bmatrix}^{\text{T}}
\end{equation}
In the above column vector, note that the coefficients are arranged in the order of increasing Hamming weight sectors. A Hamming-weight sector consists of coefficients corresponding to operators that share the same Hamming weight. The zero–Hamming-weight sector is shown in black, followed by the single–Hamming-weight sector in \textcolor{blue}{blue}, the two–Hamming-weight sector in \textcolor{orange}{orange}, the three–Hamming-weight sector in \textcolor{teal}{teal}, and the four–Hamming-weight sector in \textcolor{purple}{purple}. Within each sector, the coefficients are ordered lexicographically, i.e., according to increasing decimal values of the corresponding bit strings. Finally, note that the Hilbert Schmidt norm of $\rho$ is exactly the Euclidean norm of $\hwt{\rho}$,
\begin{equation}
    \left\langle\!\left\langle \rho| \rho \right\rangle\!\right\rangle_{h} = \sum_{i} \left|\alpha_i \right|^2 = \operatorname{Tr} \lf \rho^\dagger \rho \ri. 
\end{equation}

\section{Structure of the amplitude damping noise}
\label{ap:AD_noise_structure}
Our algorithm and main result rely strongly on the structure of the amplitude damping channel. In this section, we explicitly describe the action of the amplitude damping channel on the HW basis (or equivalently, on the computational basis). To begin, consider a single-qubit density matrix $\rho$ and its evolution under the amplitude damping channel. The Kraus operators of the channel in the computational basis are given by, 
\begin{equation}
    K_0 = \begin{pmatrix}1&0\\0&\sqrt{1-p}\end{pmatrix},
\qquad K_1 = \begin{pmatrix}0&\sqrt{p} \\ 0 & 0\end{pmatrix}.
\end{equation}
The total evolution is given by, 
\begin{equation}
    {\cal E}_p\left(\begin{bmatrix}\alpha_{0|0}&\alpha_{0|1}\\\alpha_{1|0}&\alpha_{1
    |1}\end{bmatrix}\right) = \begin{bmatrix}\alpha_{0|0}+p \alpha_{1|1} & \sqrt{1-p} \rho_{01} \\ \sqrt{1-p} \rho_{10} & (1-p) \rho_{11}\end{bmatrix}
\end{equation}
It is illustrative to see the evolution of coefficients in the Hamming weight basis. 
\begin{equation}
    \alpha_{0|0} \to \alpha_{0|0} + p \alpha_{1|1}, \quad \alpha_{0|1} \to \sqrt{1- p }\alpha_{0|1} , \quad \alpha_{0|1} \to \sqrt{1- p }\alpha_{1|0} , \quad \alpha_{1|1} \to {(1- p)}\alpha_{1|1}
\end{equation}
The key observation is that amplitude damping suppresses coefficients associated with higher Hamming-weight strings, while the coefficient of $\ketbra{0}{0}$ is re-fed by the coefficient corresponding to $\ketbra{1}{1}$. This structure directly generalizes and allows one to determine the evolution of any $n$-qubit vector expressed in the HW basis. Consider a $n$ -qubit HW basis element indexed by $i$, with hamming weight $[i]$, containing $m_i$ $\ketbra{0}{0}$. Up to a SWAP operation between qubits (which is inconsequential to our proofs),  without loss of generality, the operator $i$ in the computational basis can be represented as, 
\begin{equation}
    s_i= \tilde{s}_i \otimes \ketbra{0}{0}^{\otimes m_i},
\end{equation}
where, $\tilde{s}_i$ is a ket-bra element which has no $\ketbra{0}{0}$ terms. Since the operator $s_i$ contains $m_i$ instances of $\ketbra{0}{0}$, there are $2^{m_i}-1$ terms that re-feed into the coefficient of $s_i$. The refeeding coefficients correspond to the operators obtained by replacing a subset of the $\ketbra{0}{0}$ terms with $\ketbra{1}{1}$.  We label the indices of these operators by $j^{(r)}_k$, where $k \in [1,\binom{m_i}{t}]$ enumerates all operator elements obtained by replacing $t \in [1,m_i]$ instances of $\ketbra{0}{0}$ with $\ketbra{1}{1}$. Thus, the coefficient $\alpha_{i}$, corresponding to $s_{i}$,  evolves as 
\begin{equation}
    \alpha_i  \to \alpha_i \operatorname{Tr} \lf \mathcal{E}_p^{\otimes n} \lf s_i \ri s_i \ri + \sum_{t =1}^{m_i} \sum_{k = 1}^{\binom{m_i}{t}} \alpha_{j^{(t)}_k} \operatorname{Tr} \lf \mathcal{E}_p^{\otimes n} \lf s_{j^{(t)}_k} \ri s_i \ri
\end{equation}
Since  $\tilde{s}_i$ contains no $\ketbra{0}{0}$ terms, we have
\begin{equation}
    \label{Apeq:tilde_s_evol}
    \operatorname{Tr} \lf  \mathcal{E}_p^{\otimes n - m_i} \lf \tilde{s}_i \ri \tilde{s}_i \ri = (1 - p)^{|\tilde{s}_i|/2}.
\end{equation}
Similarly, we note, 
\begin{equation}
      \label{Apeq:ketbra_00_11_evol}
      \operatorname{Tr} \lf  \mathcal{E}_p^{\otimes m_i} \lf  \ketbra{0}{0}^{\otimes m_i - t} \otimes \ketbra{1}{1}^{\otimes t}\ri \ketbra{0}{0}^{\otimes m_i} \ri =  \lf \prod_{i = 1}^{m_i- t} \operatorname{Tr} \lf  \mathcal{E}_p \lf  \ketbra{0}{0} \ri \ketbra{0}{0} \ri\ri \lf \prod_{i = 1}^{ t} \operatorname{Tr} \lf  \mathcal{E}_p \lf  \ketbra{1}{1} \ri \ketbra{0}{0} \ri\ri = p^t.  
\end{equation}
Finally, note that any operator contributing to the re-feeding of $\alpha_i$ has the form (up to an irrelevant permutation of qubits)    
\begin{equation}
    \label{Apeq:tilde_s_form}
    s_{j^{(t)}_k} = \tilde{s}_i \otimes \ketbra{0}{0}^{\otimes m_i- t} \otimes \ketbra{1}{1}^{\otimes t}.
\end{equation}
Using eqs.~\eqref{Apeq:tilde_s_evol} and \eqref{Apeq:ketbra_00_11_evol} together with Eq.~\eqref{Apeq:tilde_s_form}, we obtain 
\begin{equation}
    \operatorname{Tr} \lf \mathcal{E}_p^{\otimes n} \lf s_{j^{(t)}_k} \ri s_i \ri =   \operatorname{Tr} \lf  \mathcal{E}_p^{\otimes n - m_i} \lf \tilde{s}_i \ri \tilde{s}_i \ri   \operatorname{Tr} \lf  \mathcal{E}_p^{\otimes m_i} \lf  \ketbra{0}{0}^{\otimes m_i - t} \otimes \ketbra{1}{1}^{\otimes t}\ri \ketbra{0}{0}^{\otimes m_i} \ri = (1-p)^{|\tilde{s}_i|/2} p^t.
\end{equation}
Thus, the coefficient evolves as
\begin{equation}
\label{Apeq:Coefficient_evolve}
       \alpha_i  \to (1 - p)^{|s_i|/2} \lf \alpha_i  + \sum_{t =1}^{m_i} \sum_{k = 1}^{\binom{m_i}{t}} \alpha_{j^{(t)}_k}  p^t  \ri, 
\end{equation}
where we have used the fact that $|s_i| = |\tilde{s_i}|$. A few critical observations are due at this point: 
\begin{enumerate}
\item The coefficients of operators with Hamming weight $s$ are damped by a prefactor $(1 - p)^{s/2}$.
\item In addition to the above damping, the coefficients of any operator  containing at least one $\ketbra{0}{0}$ term receives re-feeding contributions from the coefficients of operators in which $\ketbra{0}{0}$ is replaced by $\ketbra{1}{1}$. These coefficients carry a prefactor $p^t$, where $t$ is the number of such replacements.
\end{enumerate}

To make this even more explicit, we consider local amplitude damping acting on two qubits. For a density matrix $\rho$ written in the computational basis as 
\begin{equation}
\rho =  \left(
\begin{array}{cccc}
 \alpha _{\text{00$|$00}} & \alpha _{\text{00$|$01}} & \alpha _{\text{00$|$10}} & \alpha _{\text{00$|$11}} \\
 \alpha _{\text{01$|$00}} & \alpha _{\text{01$|$01}} & \alpha _{1|10} & \alpha _{\text{01$|$11}} \\
 \alpha _{\text{10$|$00}} & \alpha _{\text{10$|$01}} & \alpha _{10|10} & \alpha _{10|11} \\
 \alpha _{\text{11$|$00}} & \alpha _{\text{11$|$01}} & \alpha _{11|10} & \alpha _{11|11} \\
\end{array}
\right)
\end{equation}
the amplitude-damping channel acts as 
\begin{equation}
\label{ApEQ:Example}
\mathcal{E}_p^{\otimes 2} \lf \rho \ri =
\left(
\begin{array}{cccc}
     p^2 \alpha _{11|11}+p\alpha _{10|10}+p\alpha _{\text{01$|$01}}+\alpha _{\text{00$|$00}} & \sqrt{1-p } \left(p \alpha _{10|11}+\alpha _{\text{00$|$01}}\right) & \sqrt{1-p} \left(\alpha _{\text{00$|$10}}+p  \alpha _{\text{01$|$11}}\right) & (1 - p) \alpha _{\text{00$|$11}} \\
 \sqrt{1-p} \left(p  \alpha _{11|10}+\alpha _{\text{01$|$00}}\right) & (1 - p) \left(p  \alpha _{11|11}+\alpha _{\text{01$|$01}}\right) & (1 - p) \alpha_{01|10} & (1-p )^{3/2} \alpha _{\text{01$|$11}} \\
 \sqrt{1-p} \left(\alpha _{\text{10$|$00}}+p  \alpha _{\text{11$|$01}}\right) & (1 - p) \alpha _{\text{10$|$01}} & (1 - p) \left(p  \alpha _{11|11}+\alpha _{10|10}\right)& (1-p )^{3/2} \alpha _{10|11} \\
 (1- p) \alpha _{\text{11$|$00}} & (1-p )^{3/2} \alpha _{\text{11$|$01}} & (1-p )^{3/2} \alpha _{11|10} & (1 - p)^2 \alpha _{11|11} \\.
\end{array}
\right)
\end{equation}
Clearly, Eq.~\eqref{ApEQ:Example} follows the general pattern derived above in Eq.~\eqref{Apeq:Coefficient_evolve}.

\section{Proof of Proposition \ref{prop:Truncation_error}}
\label{ap:bound_based_00_terms}
In this section, we prove Proposition~\ref{prop:Truncation_error} used to prove Lemma~\ref{lem:Trunc_error} in the main text. Consider an IQP circuit composed of $d$ layers of diagonal unitary, interspersed with local amplitude damping channel on all qubits. The circuit can therefore be written as, 
\begin{equation}
   \label{Apeq:Circuit_structure}
    \mathcal{C} =  \mathcal{E}~ U_d \mathcal{E} \cdots ~U_2 \mathcal{E}U_1. 
\end{equation}
We consider the evolution of the state $\ket{+}^{\otimes n}$ under the circuit $\mathcal{C}$. In Proposition~\ref{prop:Truncation_error},  we seek to bound the coefficient of an operator $s$ in the HW basis. For an operator indexed by $i$, we denote by $\alpha_i^{(r)}$ its coefficient after $r$ layers of the circuit, where each layer consists of a diagonal unitary followed by noise. Since the initial state is $\ket{+}^{\otimes n}$, we have $\alpha_i^{(0)} = 2^{-n}$ for all $i$. The formal proof, presented in Subsection~\ref{Ap:Formal_proof}, proceeds by induction. Before turning to the general argument, however, we deem it very illustrative to gain some intuition by analyzing the evolution of the coefficients of operators containing $m_i$ factors of $\ketbra{0}{0}$, considering explicitly the cases $m_i = 0,1,2$. 
\subsection{Intuition building}
\subsubsection{\texorpdfstring{$m = 0$}{m = 0}}
In this case, no re-feeding occurs, and the noise channel only induces damping. Since the unitary gates are diagonal in the computational basis, their action is limited to introducing phase factors in the coefficients. Consequently, they do not affect the absolute values of the coefficients. Therefore, the coefficients can be tracked exactly up to an overall phase. Let $\tilde{\phi}_i^{(r)}, \phi_i^{(r)}$ denote the phase acquired by the operator indexed by $i$ at layer $r$ and the cumulative phase acquired by the coefficient corresponding to the operator indexed by $i$ after layer $r$, respectively. The evolution of the coefficients is then described by the following sequence of equations:
\begin{align}
    \alpha_i^{(0)} &= \frac{1}{2^n}, \\
    \alpha_i^{(1)} &= (1 - p)^{[i]/2}\alpha_i^{(0)} e^{i \tilde{\phi}_i^{(1)}},  \\
     &= (1 - p)^{[i]/2}\alpha_i^{(0)} e^{i \phi_i^{(1)}},  \\
    \alpha_i^{(2)} &= (1 - p)^{[i]/2}\alpha_i^{(1)} e^{i \tilde{\phi}_i^{(2)}}, \\
    &= (1 - p)^{2 [i]/2}\alpha_i^{(0)} e^{i \phi_i^{(2)}} \\
    &~ \vdots\notag  \\
    \alpha_i^{(r)} &= (1 - p)^{r[i]/2}\alpha_i^{(0)} e^{i \phi_i^{(n)}}.
\end{align}
Thus, 
\begin{equation}
    \left| \alpha_i^{(r)} \right| = \frac{(1 - p)^{r[i]/2}}{2^n}
\end{equation}
\subsubsection{\texorpdfstring{$m = 1$}{m = 1}}
In this case, re-feeding occurs from a single term. Let $j$ denote the index of the operator that re-feeds into $i$. Since $s_i$ contains exactly one $\ketbra{0}{0}$, and the operator indexed by $j$ replaces this term with $\ketbra{1}{1}$, the operator $s_j$ contains no $\ketbra{0}{0}$ terms. This further implies that $[j] = [i] + 2$. Therefore, the evolution is described by the following sequence of equations:
\begin{align}
     \alpha_i^{(0)} &= \frac{1}{2^n}, \\
    \alpha_i^{(1)} &= (1 - p)^{[i]/2}\lf \alpha_i^{(0)} e^{i \tilde{\phi}^{(1)}_i} + p \alpha_j^{(0)} e^{i \tilde{\phi}^{(1)}_j} \ri,  \\
    \alpha_i^{(2)} &= (1 - p)^{[i]/2}\lf \alpha_i^{(1)} e^{i \tilde{\phi}^{(2)}_i} + p \alpha_j^{(1)} e^{i \tilde{\phi}^{(2)}_j} \ri,\\
    &= (1 - p)^{[i]/2}\lf (1 - p)^{[i]/2}\lf \alpha_i^{(0)} e^{i \tilde{\phi}^{(1)}_i} + p \alpha_j^{(0)} e^{i \tilde{\phi}^{(1)}_j} \ri e^{i \tilde{\phi}^{(2)}_i} + p (1 - p)^{[j]/2} \alpha_j^{(0)} e^{i \phi^{(2)}_j} \ri, \\
    &=  (1 - p)^{2 [i]/2}\lf \alpha_i^{(0)} e^{i {\phi}^{(2)}_i} + p \alpha_j^{(0)} \lf e^{i {\phi}^{(2)}_{j,1}} + (1 - p) e^{i {\phi}^{(2)}_{j,2}}  \ri \ri, \\
    &~\vdots \notag \\
    \alpha_i^{(r)} &= (1 - p)^{r[i]/2}\lf \alpha_i^{(0)} e^{i {\phi}^{(r)}_i} + p \alpha_j^{(0)} \lf \sum_{a = 0}^{r - 1} e^{i{\phi}^{(r)}_{j,a}} (1 - p)^a \ri   \ri.
\end{align}
In the above equations, we use ${\phi}^{(r)}_i$ to denote the phase of the term $\alpha_i^{(0)}$  contributing to the coefficient of the operator $s_i$ after $r$ layers, and ${\phi}^{(r)}_{j,a}$ to denote the phase associated with the re-feeding contribution to the coefficient proportional to $(1-p)^a$. Applying the triangle inequality to the final line, we obtain, 
\begin{equation}
     \left| \alpha_i^{(r)} \right| \leq  \frac{(1 - p)^{r[i]/2}}{2^n} \lf 1 + p \sum_{a = 0}^{r-1} (1 - p)^a \ri = \frac{(1 - p)^{r[i]/2}}{2^n} \lf 2 - (1 - p)^r \ri. 
\end{equation}
\subsubsection{\texorpdfstring{$m = 2$}{m = 2}}
This implies that there are exactly two $\ketbra{0}{0}$ factors in the operator corresponding to index $i$. In this case, two types of re-feeding occur after each layer. The first arises from indices $j$ and $\tilde{j}$ such that $[j] = [\tilde{j}] = [i] + 2$, each contributing with coefficient $p$; this corresponds to re-feeding generated by flipping a single $\ketbra{0}{0}$ to $\ketbra{1}{1}$. The second arises from an index $k$ such that $[k] = [i] + 4$, contributing with coefficient $p^{2}$, corresponding to flipping two such terms. Note that, by definition, the term indexed by $k$ also re-feeds into the coefficients indexed by $j$ and $\tilde{j}$. Although the resulting expressions for the coefficients $\alpha_i$ become increasingly cumbersome and lack a simple closed-form pattern, they can nevertheless be upper bounded. Thus,  
\begin{align}
     \label{Apeq:intuition_m2_a}
     |\alpha_i^{(0)}| &= \frac{1}{2^n}. \\
      \label{Apeq:intuition_m2_b}
    |\alpha_i^{(1)}| &= (1 - p)^{[i]/2}\left|\lf \alpha_i^{(0)} e^{i \phi^{(1)}_i} + p \alpha_j^{(0)} e^{i \phi^{(1)}_j}+ p \alpha_{\tilde{j}}^{(0)} e^{i \phi^{(1)}_{\tilde{j}}} + p^{2} \alpha_k^{(0)} e^{i \phi^{(1)}_k} \ri\right|,  \\
     \label{Apeq:intuition_m2_c}
    &\leq  \frac{(1 - p)^{[i]/2}}{2^n} \lf 1 + 2 p  + p^2 \ri, \\
     \label{Apeq:intuition_m2_d}
    &=\frac{(1 - p)^{[i]/2}}{2^n} (1 + p)^2,  \\
    \label{Apeq:intuition_m2_e}
    |\alpha_i^{(2)}| &= (1 - p)^{[i]/2}\left \vert \lf \alpha_i^{(1)} e^{i \phi^{(2)}_i} + p \alpha_j^{(1)} e^{i \phi^{(2)}_j}+ p \alpha_{\tilde{j}}^{(1)} e^{i \phi^{(2)}_{\tilde{j}}} + p^{2} \alpha_k^{(1)} e^{i \phi^{(2)}_k} \ri \right \vert,  \\
     \label{Apeq:intuition_m2_f}
    &\leq (1 - p)^{[i]/2} \lf \left \vert \alpha_i^{(1)}\right \vert  + p  \left \vert \alpha_j^{(1)}\right \vert  +  p \left \vert \alpha_{\tilde{j}}^{(1)}\right \vert   +  p^2\left \vert \alpha_k^{(1)}\right \vert \ri, \\
    \label{Apeq:intuition_m2_g}
    & \leq \frac{(1 - p)^{[i]/2}}{2^n} \lf (1 - p)^{[i]/2} (1 + p)^2  + p  (1 - p)^{[j]/2} (1 + p) +  p  (1 - p)^{[\tilde{j}]/2} (1 + p)   +  p^2 (1 - p)^{[k]/2} \ri, \\
    \label{Apeq:intuition_m2_h}
    &=\frac{(1 - p)^{2[i]/2}}{2^n} \lf (1 + p)^2 + 2p(1-p) (1 + p) + (p(1-p))^2 \ri, \\
    \label{Apeq:intuition_m2_i}
    &= \frac{(1 - p)^{2[i]/2}}{2^n} \lf 1 + p\lf \sum_{m = 0}^{1} (1 - p)^m \ri\ri^2.
\end{align}
In the above chain of equations, we have used the triangle inequality in lines~\eqref{Apeq:intuition_m2_c} and \eqref{Apeq:intuition_m2_f}.
Following the pattern for $m=2$, we see that we get, 
\begin{equation}
    \left| \alpha_i^{(r)} \right| \leq  \frac{(1 - p)^{r[i]/2}}{2^n} \lf 1 + p \sum_{a = 0}^{r-1} (1 - p)^a \ri^2 = \frac{(1 - p)^{r[i]/2}}{2^n} \lf 2 - (1 - p)^r \ri^2.
\end{equation}
Building upon this intuition, we can show that, 
\begin{equation}
    \left| \alpha_i^{(r)} \right| \leq   \frac{(1 - p)^{r[i]/2}}{2^n} \lf 2 - (1 - p)^r \ri^{m_i}, 
\end{equation}
where $m_i$ is the number of $\ketbra{0}{0}$ in the operator indexed by $i$. We can prove this more concretely by induction. 

\subsection{Formal proof of proposition \ref{prop:Truncation_error}}
\label{Ap:Formal_proof}
\begin{prop}
\label{prop:Truncation_error}
  Consider an initial $\ket{+}^{\otimes n}$ evolved under the noisy IQP circuit $\mathcal{C}$, as defined in Thm \ref{thm:main_theorem}. Let $i$ denote the index, in the Hamming-weight basis, of an operator containing $m_i$ tensor factors of $\ketbra{0}{0}$, and let $\alpha_i^{(r)}$ denote the corresponding coefficient after $r$ layers of the noisy IQP circuit. Then, 
  \begin{equation}
      \left| \alpha_i^{r} \right| \leq \frac{(1 - p)^{r[i]/2}}{2^n} \big(2 - (1 - p)^r\big)^{m_i}.
      \label{Eq:Coeff_bound}
  \end{equation}
  Furthermore, this bound is tight for all $\alpha_i^{(r)}$ as it can be saturated when all unitary gates are identities. 
\end{prop}

\textbf{Proof:} First note that the above proposition is trivially satisfied for all $i$ for $r = 0$. Suppose it holds for all indices at layer $k$. Consider an arbitrary index $i$ and depth $k+1$. Define the set $\mathcal{R}_j \subset \{1,\dots,n\}$ as the set of indices labeling operators obtained by replacing exactly $j$ factors of $\ketbra{0}{0}$ in the operator indexed by $i$, with $\ketbra{1}{1}$. $n$ is the total number of qubits. These operators feed into $\alpha_i$ at the order of $p^j$. Let $m_i$ denote the number of $\ketbra{0}{0}$'s in the operator indexed by $i$. Then the coefficient $\alpha_i^{(k+1)}$ is given as 
\begingroup
\setlength{\abovedisplayskip}{6pt}
\setlength{\belowdisplayskip}{6pt}
\setlength{\abovedisplayshortskip}{6pt}
\setlength{\belowdisplayshortskip}{6pt}
\begin{equation}
        \alpha_i^{(k+1)} = (1 - p)^{[i]/2} \lf \alpha_i^{(k)} e^{i \phi_0} + \sum_{a = 1}^{m_i} p^a \lf \sum_{b_a \in \mathcal{R}^{(i)}_a} \alpha_{b_a}^{(k)} e^{i \phi_{b_a}} \ri \ri.
\end{equation}
In the above equation, $e^{i \phi_0}$ is the phase accumulated by the operator indexed by $i$ at layer $k+1$. Similarly,  $e^{i \phi_{b_a}}$ is the phase accumulated by the operator indexed by $b_a$ at layer $k+1$, where $b_a$ indexed an operator in $\mathcal{R}_a$. Applying the triangle inequality, we get the following chain of inequalities, 
\begin{align}
     \label{Apeq:Bound_Line1}
    \left| \alpha_i^{(k+1)} \right|&\leq (1 - p)^{[i]/2}  \lf \left|\alpha_i^{(k)} \right| + \sum_{a = 1}^{m_i} p^a  \sum_{{b_a} \in \mathcal{R}^{(i)}_a} \left| \alpha_{b_a}^{(k)} \right|  \ri, \\
    \label{Apeq:Bound_Line2}
     &\leq \frac{(1 - p)^{[i]/2}}{2^n} \lf (1 - p)^{k[i]/2} \lf 2 - (1 - p)^k \ri^{m_i} + \sum_{a = 1}^{m_1} p^a  \sum_{{b_a} \in \mathcal{R}^{(i)}_a} (1 - p)^{k[{b_a}]/2} \lf 2 - (1 - p)^k \ri^{m_{b_a}}  \ri, \\
     \label{Apeq:Bound_Line3}
    &= \frac{(1 - p)^{(k + 1)[i]/2}}{2^n} \lf \lf 2 - (1 - p)^k \ri^{m_i} + \sum_{a = 1}^{m_i} p^a (1 - p)^{ka}    \lf 2 - (1 - p)^k \ri^{m_i - a} \sum_{{b_a} \in \mathcal{R}^{(i)}_a} \ri, \\
    \label{Apeq:Bound_Line4}
    &= \frac{(1 - p)^{(k + 1)[i]/2}}{2^n} \lf \lf 2 - (1 - p)^k \ri^{m_i} + \sum_{a = 1}^{m_i} p^a (1 - p)^{ka}    \lf 2 - (1 - p)^k \ri^{m_i - a} {\binom{m_i}{a}} \ri, \\
    \label{Apeq:Bound_Line5}
    &= \frac{(1 - p)^{(k + 1)[i]/2}}{2^n}  \lf  2 - (1 - p)^k  + p(1 - p)^k\ri^{m_i}, \\
    \label{Apeq:Bound_Line6}
    &= \frac{(1 - p)^{(k + 1)[i]/2}}{2^n}  \lf  2 - (1 - p)^{k+1}\ri^{m_i} .
\end{align}  
\endgroup
In the above chain of inequalities, the triangle inequality is applied in line~\eqref{Apeq:Bound_Line1}, and the induction hypothesis is invoked in line~\eqref{Apeq:Bound_Line2}. By the definition of the set $\mathcal{R}^{(i)}_a$, for every operator indexed by ${b_a} \in \mathcal{R}^{(i)}_a$, the number of factors of $\ketbra{0}{0}$ satisfies $m_{b_a} = m_i - a$, and the Hamming weight satisfies $[{b_a}] = [i] + 2a$. These relations are used in line~\eqref{Apeq:Bound_Line3}. The cardinality of the set $\mathcal{R}^{(i)}_a$, given by $|\mathcal{R}^{(i)}_a| = \binom{m_i}{a}$, is used in line~\eqref{Apeq:Bound_Line4}. This is easy to see as the set $\mathcal{R}^{(i)}_a$ contains all operators where $a$ out of $m_i$ factors of $\ketbra{0}{0}$ are replaced.  Finally, the binomial expansion is applied in line~\eqref{Apeq:Bound_Line5}, and the resulting expression is rearranged in line~\eqref{Apeq:Bound_Line6}. Thus, we have proved the above proposition by induction. Note that if all the coefficients are positive reals, $\alpha^{(k)}_i,\alpha^{(k)}_{b_a} \geq 0$,and if all the phases were unity, $e^{i \phi_0},e^{i{\phi_{b_a}}}=1$, then the triangle inequality is saturated. Hence, if all the diagonal gates are identity, Eq.~\eqref{Eq:Coeff_bound} is saturated. Thus, the bound is tight.

\section{Maximum and minimum \texorpdfstring{$\ketbra{0}{0}$}{|0><0|} possible for a strings with a given hamming weight}
\label{ap:Counting_argument}
In this section, we task ourselves with answering the following questions: For a given hamming weight $h$ what is the total number of terms with $z$ $\ketbra{0}{0}$'s. We will do this by starting with the maximum possible $z$ and go down to the minimum. For clarity, we will separate even and odd $h$.

\subsection{Even \texorpdfstring{$h$}{h}}

Since the hamming weight is $h$, we need to have $h$ ones and $2n - h$ zeros in our operator. The maximum number of $\ketbra{0}{0}$ is when all the ones are paired up. Let $\mu_h$ denote the maximum number of $\ketbra{0}{0}$ in an operator with a hamming weight $h$. By the above argument, 
\begin{equation}
    \mu_h = n - \frac{h}{2}.
\end{equation}

To count the number of such operators, note that out of $n$ qubits we have $\mu_h$ $\ketbra{0}{0}$ and the remaining are $\ketbra{1}{1}$. Let $t_{z_h}$ count the number of operators with $z_h~ \ketbra{0}{0}$'s with a hamming weight $h$. Thus, 
\begin{equation}
    t_{\mu_h} = {\binom{n}{\mu_h}} = {\binom{n}{h/2}}.
\end{equation}
Now consider the number of operators with $\mu_h - 1$ $\ketbra{0}{0}$'s. To count this, note that such operators can be formed by breaking up one $\ketbra{0}{0}$ and $\ketbra{1}{1}$ and recombining them into $\ketbra{0}{1}$ or $\ketbra{1}{0}$. Thus out of $n$ qubits now we have, $h/2 - 1$ $\ketbra{1}{1}$'s, $\mu_h - 1$ $\ketbra{0}{0}$'s and the remaining $2$ position has $\ketbra{0}{1}$ or $\ketbra{1}{0}$. Thus, 
\begin{equation}
    t_{\mu_h - 1} = {\binom{n}{h/2 + 1}}{\binom{h/2 + 1}{2}} 2^2.
\end{equation}
The factor of 2 is because we can have either $\ketbra{0}{1}$ or $\ketbra{1}{0}$ in each of the two spots. Continuing the pattern, observe that, 
\begin{equation}
    t_{\mu_h - 2} = {\binom{n}{h/2 + 2}}{\binom{h/2 + 1}{4}} 2^4.
\end{equation}
Thus, after $r$ reductions, 
\begin{equation}
    t_{\mu_h - r} = {\binom{n}{h/2 + r}}{\binom{h/2 + 1}{2r}} 2^{2r}.
\end{equation}
The minimum number of $\ketbra{0}{0}$ for a given hamming weight $h$, $l_h$ is either when all $\ketbra{0}{0}$ have been converted or all $\ketbra{1}{1}$ have been converted, whichever happens first. In the first case, we will have $l_h = 0$ and in the second case $l_h = n - h/2 - h/2 = n - h$. Thus, 
\begin{equation}
    l_h = \max(0,n - h).
\end{equation}

\subsection{Odd \texorpdfstring{$h$}{h}}

As before, since the hamming weight is $h$, we need to have $h$ ones and $2n - h$ zeros in our operator. The maximum number of $\ketbra{0}{0}$ is when the maximum number of ones are paired up, which would be all but 1. Let $\mu_h$ denote the maximum number of $\ketbra{0}{0}$ in an operator with a hamming weight $h$. By the above argument, 
\begin{equation}
    \mu_h = n - \frac{h+1}{2}.
\end{equation}
To count the number of such operators, note that out of $n$ qubits we have $\mu_h$ $\ketbra{0}{0}$, one one is either $\ketbra{0}{1}$ or $\ketbra{1}{0}$,  and the remaining are $\ketbra{1}{1}$. Let $t_{a_h}$ count the number of operators with $a~ \ketbra{0}{0}$'s with a hamming weight $h$. Thus, 
\begin{equation}
    t_{\mu_h} = {\binom{n}{(h+1)/2}} {\binom{(h+1)/2}{1}}2.
\end{equation}
Now consider the number of operators with $\mu_j - 1$ $\ketbra{0}{0}$'s. To count this, note that such operators can be formed by breaking up one $\ketbra{0}{0}$ and $\ketbra{1}{1}$ and recombining them into $\ketbra{0}{1}$ or $\ketbra{1}{0}$. Thus out of $n$ qubits now we have, $(h-1)/2 - 1$ $\ketbra{1}{1}$'s, $\mu_h - 1$ $\ketbra{0}{0}$'s and the remaining $3$ position has $\ketbra{0}{1}$ or $\ketbra{1}{0}$. Thus, 
\begin{equation}
    t_{\mu_h - 1} = {\binom{n}{(h + 1)/2 + 1}}{\binom{(h + 1)/2 + 1}{3}} 2^3.
\end{equation}
Thus, after $r$ reductions, 
\begin{equation}
    t_{\mu_h - r} = {\binom{n}{(h + 1)/2 + r}}{\binom{(h + 1)/2 + r}{2r + 1}} 2^{2r + 1}.
\end{equation}
The minimum number of $\ketbra{0}{0}$ for a given hamming weight $h$, $l_h$ is when either all $\ketbra{0}{0}$ have been converted or all $\ketbra{1}{1}$ have been converted, whichever happens first. In the first case, we will have $l_h = 0$ and in the second case $l_h = n - (h + 1)/2 - (h - 1)/2 = n - h$. Thus, 
\begin{equation}
    l_h = \max(0,n - h).
\end{equation}

Finally, combining the even and the odd cases, for strings with a given hamming weight $h$, the  maximum (minimum) number of $\ketbra{0}{0}$, denoted by $\mu_h$ ($l_h$) possible is given by, 
\begin{equation}
    \mu_h = n - \left\lceil\frac{h}{2}\right \rceil, \quad l_h = n - h .
\end{equation}

\section{Proof of Lemma 1}
\label{ap:Tail_bound}
In this section, we prove Lemma~\ref{lem:Trunc_error}, from the main text. For the sake of completeness, we restate the Lemma below, 
\begin{lemma*}[\textbf{\ref{lem:Trunc_error}}]
    Using the same notation as in Theorem \ref{thm:main_theorem}, there exists a depth threshold $d^*_T = \log(n)/\log(1-p)^{-1}$, such that for $d > d^*_T$, the truncation error in Hilbert-Schmidt distance, incurred by truncating operators with Hamming weight greater than $k$, is upperbounded by 
    \begin{equation}
        \varepsilon^2   \leq \frac{\lf 2 - (1 - p)^d \ri^{2n - k - 1}}{4^n} e^{2n H\lf\frac{k+1}{2n} \ri} (1 - p)^{d(k+1)},\label{Apeq:hilbert-schmidt-error}
    \end{equation}
\end{lemma*}
\textit{Proof:} We upper bound the truncation error using Chernoff bounds \cite{chernoff1952measure}. For an operator indexed by $i$ in the HW basis, we denote by $\alpha_i^{(r)}$ its coefficient after $r$ layers of the circuit, where each layer consists of a diagonal unitary followed by noise. Consider the Hilbert-Schmidt error, that we defined in Eq.~\eqref{eq:HS_error} in the main text, 
\begin{equation}
     \varepsilon^2 := \|\rho - \sigma\|_{H.S}  = \left\langle\!\left\langle \rho - \sigma| \rho - \sigma \right\rangle\!\right\rangle_{h} = \sum_{i \in \mathcal{T}} \left|\alpha_i^{(d)}\right|^2.
\end{equation}
 The following chain of equations follow, 
\begin{align}
    \varepsilon^2 &= \sum_{i \in \mathcal{T}} \left| \alpha_i^{(d)} \right|^2, \label{Eq:trunc_line1}\\
    &= \sum_{h = k+1}^{2n} \sum_{\substack{i \in \mathcal{T} \\ [i] = h}} \left| \alpha_i^{(d)} \right|^2, \label{Eq:trunc_line2}\\
    &= \sum_{h = k+1}^{2n} \sum_{r =  l_h}^{\mu_h } \sum_{\substack{i \in \mathcal{T} \\ [i] = h \\ `r' \ketbra{0}{0}s.}} \left| \alpha_i^{(d)} \right|^2, \label{Eq:trunc_line3}\\
    &\leq \sum_{h = k+1}^{2n}  \sum_{r =  l_h}^{\mu_h } \frac{(1 - p)^{dh}}{4^n} \lf 2 - (1 - p)^d \ri^{2r} \lf  \sum_{\substack{i \in \mathcal{T} \\ [i] = h \\ `r' \ketbra{0}{0}s.}}  \ri, \label{Eq:trunc_line4}\\
    &= \sum_{h = k+1}^{2n} \frac{(1 - p)^{dh}}{4^n} \sum_{r = l_h}^{\mu_h}  \lf 2 - (1 - p)^d \ri^{2 r}t_{r}, \label{Eq:trunc_line5}\\
    &\leq \sum_{h = k+1}^{2n}  \frac{(1 - p)^{dh}}{4^n} \lf 2 - (1 - p)^d \ri^{2 \mu_h} \sum_{r = l_h}^{\mu_h}  t_{r}, \label{Eq:trunc_line6}\\
    &\leq \sum_{h = k+1}^{2n}  \frac{(1 - p)^{dh}}{4^n} \lf 2 - (1 - p)^d \ri^{2n - 2 \left \lceil \frac{h}{2}  \right \rceil} {\binom{2n}{h}}, \label{Eq:trunc_line7}\\
    & \leq \frac{\lf 2 - (1 - p)^d \ri^{2n - k - 1}}{4^n} e^{2n H\lf\frac{k+1}{2n} \ri} (1 - p)^{d(k+1)}, \label{Eq:trunc_line8}\\
    & \leq \lf 2 - (1 - p)^d \ri^{- (k + 1)} e^{2n H\lf\frac{k+1}{2n} \ri} (1 - p)^{d(k+1)}.
 \end{align}
Going from line \eqref{Eq:trunc_line1} to \eqref{Eq:trunc_line2}, we separate out terms with their Hamming weights. Going from \eqref{Eq:trunc_line2} to \eqref{Eq:trunc_line3}, we separate out terms with a given number of $\ketbra{0}{0}$'s. $\mu_{h}(l_h)$ denote the maximum (minimum) number of $\ketbra{0}{0}$ possible for a string of hamming weight $h$. The explicit form for these are given in appendix \ref{ap:Counting_argument}. We use the bound proved in Proposition \ref{prop:Truncation_error} in line \eqref{Eq:trunc_line4}. In line \eqref{Eq:trunc_line6} we use $\lf 2 - (1 - p)^d \ri^{2 r} \leq \lf 2 - (1 - p)^d \ri^{2 \mu_h}$ for $r \leq \mu_h$. In the line \eqref{Eq:trunc_line7}, we have used the fact that $\sum_{r = l_h}^{\mu_h}  t_{r}$ represents the total number of operators of a given hamming weight $h$. This is equal to the number of terms with $h$ ones inserted in $2n$ places, i.e,  ${\binom{2n}{h}}$. In the same line, we have also used the fact that the maximum number of $\ketbra{0}{0}$ factors in a term with hamming weight is achieved if all the ones (or all but one) are paired up as $\ketbra{1}{1}$. The total number of $\ketbra{1}{1}$ factors then is given by $\lceil h/2 \rceil$. Thus, $\mu_h = n - \lceil h/2 \rceil$.  In the line \eqref{Eq:trunc_line8}, $H(x) = -x \log x - (1-x) \log (1-x) $ is the binary entropy, with the logarithm being natural logarithms. Line \eqref{Eq:trunc_line7} to line \eqref{Eq:trunc_line8} is done using Chernoff tail bounds and is presented below. The use of Chernoff bound necessitates that $k \geq  n(1 - p)^d  - 1$. In equation \eqref{Eq:trunc_line7}, we wish to bound, 
\begin{equation*}
    \psi^{n,k}(p) := \sum_{h = k+1}^{2n}  \frac{(1 - p)^{dh}}{4^n} \lf 2 - (1 - p)^d \ri^{2n - 2 \left \lceil \frac{h}{2}  \right \rceil} {\binom{2n}{h}}.
\end{equation*}

Using the fact that $x \leq \lceil x\rceil$, which implies that $a^{-x} \geq a^{- \lceil x \rceil}$, we obtain, 
\begin{align}
     \psi^{n,k}(p) &\leq \frac{1}{4^n}\sum_{h = k+1}^{2n}  (1 - p)^{dh} \lf 2 - (1 - p)^d \ri^{2n - h} {\binom{2n}{h}}, \\
     &= \sum_{h = k+1}^{2n} \lf \frac{(1 - p)^{d}}{2}\ri^h \lf \frac{2 - (1 - p)^d }{2}\ri^{2n-h} {\binom{2n}{h}}, \\
    & =\phi^{(n,k)}(a), \\
\end{align}
where we have defined $a,b$ and $\phi^{(n,k)}(a)$ as, 
\begin{equation}
    a := \lf \frac{2 - (1 - p)^d }{2}\ri, \quad b:= \lf \frac{(1 - p)^{d}}{2}\ri, \quad \phi^{n,k}(a) := \sum_{h = k+1}^{2n} a^{2n - h} b^h  {\binom{2n}{h}}.
\end{equation}
Note that $\phi^{(n,k)}(a)$ is the upper-tail probability of a binomial distribution with $2n$ independent Bernoulli trials, each yielding $0$ with probability $a$ (and $1$ with probability $b$). An upper bound on this tail can be obtained via a standard Chernoff bound, derived by applying Markov’s inequality to the moment-generating function of the associated random variable \cite{chernoff1952measure}. Let $X$ denote this random variable and fix $r \geq 0$. The expectation value of the moment generating function is given by, 
\begin{equation}
    \mathbb{E}(e^{rX}) = \sum_{h = 0}^{2n} e^{rh} a^{2n - h} b^h {\binom{2n}{h}}  = (be^r + a )^{2n}.
\end{equation}
Define $\alpha := (k+1)/2n$. Applying Markov's inequality to $e^{rX}$, for any $r \geq 0$, yields
\begin{align}
   \phi^{(n,k)}(a) &= P(X \geq 2n\alpha), \label{Eq:phibd_line1} \\
   &\leq e^{-2rn\alpha}\mathbb{E}(e^{rX})),\label{Eq:phibd_line2} \\
   &= e^{-2rn\alpha} \lf a + b e^r \ri^{2n}, \label{Eq:phibd_line3} \\
   &= e^{-2rn\alpha + 2n \ln \lf a +  be^r \ri}. \label{Eq:phibd_line4}
 \end{align}
 The tightest bound is obtained by minimizing the above expression over $r$. Since the exponential function is monotonic, this is equivalent to minimizing the function in the exponent. Define 
\begin{equation*}
    f(r) = -2r n \alpha + 2n \ln \lf a + b e^r  \ri.
\end{equation*}
Let $r^*$ denote the minimizer of $f(r)$. Imposing $f'(r^*) = 0$, together with $a = 1-b$, gives 
\begin{equation}
    \label{Eq:Chernoff_sat_cond}
    e^{r^*} = \frac{(1 - b)\alpha}{(1 - \alpha)b}.
\end{equation}
To ensure that the bound is valid, we require $r^* \geq 0$, which implies $e^{r^*} \geq 1$. This holds by imposing the condition 
\begin{equation}
    \alpha \geq b \implies  k+1 \geq   n(1 - p)^d .
\end{equation}
Not imposing the above condition, the tightest upper bound would be the trivial bound $\phi^{(n,k)}(a) \leq 1$. Substituting Eq.~\eqref{Eq:Chernoff_sat_cond} in Eq.~\eqref{Eq:phibd_line4}, we get the following chain of equations. 
\begin{align}
     \phi^{(n,k)}(a) &\leq \exp \lf -2n\alpha \ln \frac{(1 - b)\alpha}{(1 -\alpha)b} + 2n \ln \lf (1 - b) + \frac{\alpha (1 - b)}{(1 - \alpha)} \ri \ri, \\
     &= e^{- 2n \alpha \ln (1 - b) - 2n \alpha \ln \alpha + 2n \alpha \ln (1 - \alpha) + 2n \alpha \ln b +  2n \ln (1 - b) - 2n \ln(1 - \alpha)}, \\
     &= e^{2n H(\alpha)} e^{2n \alpha \ln b}e^{2n (1 - \alpha)\ln (1 - b)}, \\
     &= e^{2n H(\alpha)} b^{k+1} a^{2n - k - 1}, \\
     &= \frac{\lf 2 - (1 - p)^d \ri^{2n - k - 1}}{4^n} (1 - p)^{d(k+1)} e^{2n H \lf \frac{k+1}{2n} \ri}
\end{align}
Thus, 
\begin{equation}
    \psi^{(n,k)}(p) \leq \frac{\lf 2 - (1 - p)^d \ri^{2n - k - 1}}{4^n} (1 - p)^{d(k+1)} e^{2n H \lf \frac{k+1}{2n} \ri}
\end{equation}

 \section{Bounding the trace distance with Hilbert  Schmidt distance}
\label{ap:Td_upper_bound}
In this section, we will prove Theorem~\ref{thm:Bound_on_TD} in the main text, where we upper bound the trace distance by the Hilbert-Schmidt distance when one of the inputs is a low-rank Hermitian matrix, and the other is a valid density matrix. This result is a minor generalization of the result in Ref.~\cite{coles_cerezo_bound_2019} and follows most of the same steps laid out in that reference. First, we prove the following lemma, 
\begin{lemma}
    \label{Aplem:TD_bd_lemma}
    Let $\Delta = \sigma - \rho$, where $\sigma$ is a Hermitian matrix and $\rho$ is a density matrix. Let $\Delta = \Delta_+ - \Delta_-$, where $\Delta_+$ and $\Delta_-$ correspond to the positive and the negative parts of the eigenspectrum of $\Delta$, with $\Delta_+, \Delta_- \geq 0$, and $\Delta_+ \Delta_- = \Delta_- \Delta_+ = 0$. Then, we have 
    \begin{equation}
        \operatorname{rank}(\Delta_+) \leq \operatorname{rank}(\sigma).
    \end{equation}
\end{lemma}
\textbf{Proof}: Suppose $A \in M_{n}\lf \mathbb{C} \ri$ and $B \in M_{n}\lf \mathbb{C} \ri$ are two Hermitian matrices with eigenvalues ${a_i}$ and ${b_i}$. Furthermore, suppose that the eigenvalues are sorted in decreasing order, i.e, 
\begin{equation}
    a_1 \geq \dots \geq a_n, \quad b_1 \geq \dots \geq b_n.
\end{equation}
Let $\lambda_i$ denote the eigenvalues of $A + B$ that are sorted in descending order. Then, by Weyl's inequalities \cite{bhatia2013matrix}, 
\begin{align}
    \label{Apeq:Weyl_a}
    \lambda_{j} &\leq a_i + b_{j - i +1}   ~~~\text{ for }  i  \leq j, \\ 
    \label{Apeq:Weyl_b}
    \lambda_{j} &\geq a_i + b_{j - i + n}  ~~~\text{ for } i \geq j.
\end{align}
 Let $\{\rho_j\}, \{\Delta_j\}$, and $\{\sigma_j\}$ denote the eigenvalues of $\rho, \Delta$, and $\sigma$ sorted in descending order. Let $n$ be the dimension of above matrices. Decomposing $\sigma$ into its positive and negative eigenspectrum, $\sigma = \sigma_+ - \sigma_-$, 
\begin{equation}
    \operatorname{rank}(\sigma) = \operatorname{rank}(\sigma_+) + \operatorname{rank}(\sigma_-) \geq \operatorname{rank}(\sigma_+). 
\end{equation}
Moreover, since the eigenvalues $\{\sigma_j\}$ are sorted in descending order, 
\begin{equation}
    \sigma_j > 0, \text{ for } 1 \leq j \leq \operatorname{rank}(\sigma_+) \quad \text{and} \quad \sigma_j \leq 0 , \text{ for } \operatorname{rank}(\sigma_+) < j \leq n .
\end{equation}
Applying Eq.~\eqref{Apeq:Weyl_b} to $\sigma = \rho + \Delta$, when $j = i$, we get 
\begin{equation}
    \sigma_i \geq \Delta_i + \rho_n ~~~ \text{for all } i. 
\end{equation}
Since $\rho$ is a density matrix, it is positive $\rho \geq 0$. Thus,  $\rho_n \geq 0$ implies that
\begin{equation}
    \sigma_i \geq \Delta_i~~~ \text{for all } i.
\end{equation}
Therefore, when $\sigma_{i} \leq 0 $, we are guaranteed that $\Delta_i \leq 0$. This implies that
\begin{equation}
    i \geq \operatorname{rank}(\sigma_+) \implies i \geq \operatorname{rank}(\Delta_+).
\end{equation}
Thus, 
\begin{equation}
    \operatorname{rank}(\sigma) \geq  \operatorname{rank}(\sigma_+)  \geq \operatorname{rank}(\Delta_+).
\end{equation}

For the sake of clarity, we restate Theorem~\ref{thm:Bound_on_TD} in the main text here.
\begin{thm*}[\textbf{\ref{thm:Bound_on_TD}}]
    Suppose $\rho$ is a density matrix and $\sigma$ a non-zero hermitian matrix. Suppose $\operatorname{Tr}(\sigma) \leq 1$ and there exists an $\epsilon' \geq 0$ such that 
    \begin{align}
        &|\operatorname{Tr}(\rho - \sigma) | \leq \epsilon',  \\
        & \|\rho -\sigma \|_{\text{HS}} \leq \epsilon'.
    \end{align}
    Then, we can show that
    \begin{equation}
        \label{Apeq:Bound_eq}
        \|\rho -\sigma \|_{\text{TD}} \leq \left(\sqrt{\operatorname{rank}(\sigma)} + 1 \right)\epsilon'.
    \end{equation}
\end{thm*}
\textbf{Proof:} Define $\Delta = \sigma - \rho$, $\Delta_{+}$, and $\Delta_-$ as in Lemma~\ref{Aplem:TD_bd_lemma}. Both the trace distance and the Hilbert-Schmidt distance have a nice form in terms of  $\Delta_{+}$ and $\Delta_-$,
\begin{align}
    \label{Apeq:HS_defn}
     \|\rho -\sigma \|_{\text{HS}} &=  \|\Delta \|_{\text{HS}} = \sqrt{\operatorname{Tr} \lf \Delta_+^2 \ri + \operatorname{Tr} \lf \Delta_-^2 \ri}, \\
     \label{Apeq:TD_defn}
     \|\rho -\sigma \|_{\text{TD}} &=  \|\Delta \|_{\text{TD}} = \frac{1}{2} \left[\operatorname{Tr} \lf \Delta_+ \ri + \operatorname{Tr} \lf \Delta_- \ri \right].
\end{align}
We will prove this theorem by considering three different cases:
\begin{enumerate}
    \item $\Delta_- = 0$,
    \item  $\Delta_- \neq 0$ and $\Delta_+ = 0$,
    \item  $\Delta_- \neq 0$ and $\Delta_+ \neq 0$.
\end{enumerate}
\par\textbf{Case 1:} We start by noting that $\Delta_- = 0$ implies that $\rho = \sigma$. To see this, consider the $\operatorname{Tr}(\sigma)$ under the assumption that $\Delta_- = 0$, 
\begin{equation}
    \operatorname{Tr}(\sigma) = \operatorname{Tr}(\rho) + \operatorname{Tr}(\Delta) = 1 + \operatorname{Tr}(\Delta_+).
\end{equation}
The condition $\operatorname{Tr(\sigma)} \leq 1$ and the above equation can only be satisfied if $\operatorname{Tr}(\Delta_+) = 0$. Since by construction $\Delta_+ \geq 0$, this implies that $\Delta_+ = 0$ and hence $\rho = \sigma$. This implies that Eq.\eqref{Apeq:Bound_eq} is satisfied in this case.  

\textbf{Case 2:} When $\Delta_+ = 0$, we have 
\begin{equation}
    \rho - \sigma =  \Delta_-.
\end{equation}
From the assumptions in the theorem, 
\begin{equation}
    \operatorname{Tr}(\Delta_-)  =  \operatorname{Tr}(\rho - \sigma)= |\operatorname{Tr}(\sigma - \rho)|  \leq \epsilon'.
\end{equation}
Thus, the trace distance can be written as 
\begin{equation}
    2\|\rho -\sigma \|_{\text{TD}} = \operatorname{Tr}(\Delta_+)  + \operatorname{Tr}(\Delta_-)  = \operatorname{Tr}(\Delta_-)\leq \epsilon '\leq  \lf \sqrt{\operatorname{rank}(\sigma)}  + 1\ri  \epsilon'.
\end{equation} 

\textbf{Case 3:} Here both $\Delta_+$ and $\Delta_-$ are nonzero.  Using the same proof technique as in Ref \cite{coles_cerezo_bound_2019}, define the operators 
\begin{equation}
    \label{Apeq:tau_def}
    \tau_+ = \frac{\Delta_+}{\operatorname{Tr}\lf \Delta_+ \ri}, \quad \tau_- = \frac{\Delta_-}{\operatorname{Tr}\lf \Delta_- \ri}. 
\end{equation}
Note that $\tau_\pm$ are valid density matrices. The purity of a density matrix is lower bounded by the inverse of its rank. Using the Lemma~\ref{Aplem:TD_bd_lemma}, 
\begin{equation}
     \label{Apeq:another_eq_2_label}
    \operatorname{Tr} \lf \tau_+^2 \ri \geq \frac{1}{\operatorname{rank}(\Delta_+)} \geq \frac{1}{\operatorname{rank}(\sigma)}.
\end{equation}
Thus, using the definition of $\tau_+$ in Eq.~\eqref{Apeq:tau_def} and Eq.~\eqref{Apeq:another_eq_2_label},
\begin{equation}
    \operatorname{Tr}(\Delta_+)^2 \geq \frac{\operatorname{Tr} \lf \Delta_+ \ri^2}{\operatorname{rank}(\sigma)}
\end{equation}
Similarly, we also get
\begin{equation}
    \operatorname{Tr} \lf \Delta_-^2 \ri \geq \frac{\operatorname{Tr}(\Delta_-)^2}{\operatorname{rank}(\Delta_-)}. 
\end{equation}
Adding the above two equations, we get 
\begin{equation}
    \label{Apeq:another_eqn}
    \operatorname{Tr} \lf \Delta_+^2 \ri + \operatorname{Tr} \lf \Delta_-^2 \ri \geq \frac{\operatorname{Tr} \lf \Delta_+ \ri^2}{\operatorname{rank}(\sigma)} + \frac{\operatorname{Tr} \lf \Delta_- \ri^2}{\operatorname{rank}(\Delta_-)}
\end{equation}
Define $\delta$ as
\begin{equation}
    \label{Apeq:Delta_defn}
    \delta = \frac{\operatorname{Tr} \lf \rho - \sigma \ri}{2} = \frac{\operatorname{Tr}(\Delta_- - \Delta_+)}{2}
\end{equation}
Combining Eqs.~\eqref{Apeq:TD_defn} and \eqref{Apeq:Delta_defn}  
\begin{equation}
   \label{Apeq:anotheranother_eqn}
  \operatorname{Tr}(\Delta_-) =  \|\rho -\sigma \|_{\text{TD}} + \delta , \quad \operatorname{Tr}(\Delta_+) =  \|\rho -\sigma \|_{\text{TD}}- \delta .
\end{equation}
Using Eqs.~\eqref{Apeq:HS_defn}, \eqref{Apeq:another_eqn} and \eqref{Apeq:anotheranother_eqn}, we get, 
\begin{equation}
     \|\rho -\sigma \|_{\text{HS}}^2 \geq \frac{( \|\rho -\sigma \|_{\text{TD}} - \delta)^2}{\operatorname{rank}(\sigma)} + \frac{( \|\rho -\sigma \|_{\text{TD}} + \delta)^2}{\operatorname{rank}(\Delta_-^2)} \geq \frac{( \|\rho -\sigma \|_{\text{TD}} - \delta)^2}{\operatorname{rank}(\sigma)}.
\end{equation}
Thus, 
\begin{equation}
     \|\rho -\sigma \|_{\text{TD}}  \leq \sqrt{\text{rank}(\sigma)  \|\rho -\sigma \|_{\text{HS}}^2} + \delta
\end{equation}
Finally, using the assumptions $ \|\rho -\sigma \|_{\text{HS}}\leq \epsilon'$ and $\delta \leq \epsilon'$, we have 
\begin{equation}
     \|\rho -\sigma \|_{\text{TD}} \leq \lf \sqrt{\operatorname{rank}(\sigma)}  + 1\ri \epsilon'
\end{equation}

\section{Evaluating the coefficients of operator strings under noisy IQP circuit}
\label{Ap:Low_Hamming_Weight_Strings}
In this section, we present an explicit algorithm for computing the evolution of the coefficient $\alpha_i$ of an operator indexed by $i$ in the HW basis.  As seen in Appendix~\ref{ap:AD_noise_structure}, amplitude damping decays the coefficients of HW basis operators according to their Hamming weight and refeeds coefficients from $\ketbra{1}{1}$ to $\ketbra{0}{0}$. Equivalently, the channel damps the off-diagonal elements ($\ketbra{0}{1}, \ketbra{1}{0}$) while mixing the diagonal basis elements ($\ketbra{0}{0},\ketbra{1}{1}$). In contrast, the diagonal gates act non-trivially only on the off-diagonal elements.  

Motivated by this structure, rather than directly working in the HW basis, we introduce a new over-complete operator basis that groups the diagonal basis elements together. Define  
\begin{equation}
    I(a) = \ketbra{0}{0} + a \ketbra{1}{1}.
\end{equation}
Let $\sigma_{\pm}$ be defined as the usual raising and lowering operators on a single qubit,
\begin{equation}
    \sigma_+ = \ketbra{1}{0},\qquad \sigma_{-} = \ketbra{0}{1}.
\end{equation}
Then, 
\begin{equation}
    \mathcal{I} := \{I(p), \sigma_+, \sigma_-| p \in \mathbb{C}, 0 \leq |p| \leq 1\}
\end{equation}
forms an operator frame (over-complete operator basis). This representation is particularly convenient for noisy IQP circuits, as it allows us to exactly keep track of the evolution of the elements of the frame.  To see that, we start by expressing the initial state in this frame. Using the expansion $\ketbra{+}{+} = 1/2(I(1) + \sigma_+ + \sigma_-)$, we write the initial state as a sum of operator strings $s \in \mathcal{I}^{\otimes n}$, 
\begin{equation}
\label{Apeq:Initial_state_expansion}
    \ketbra{+}{+}^{\otimes n} = \frac{1}{2^n} \lf I(1) + \sigma_+ + \sigma_-\ri^{\otimes n} = \frac{1}{2^n} \sum_{r = 0}^{n} ~\sum_{\pi \in \mathbb{P}^r_n}~ \sum_{s_1, \dots, s_r = +, -}   ~ \sigma_{s_1,\pi(1)}  \cdots \sigma_{s_r,\pi(r)} I(1)_{\pi(r+1)}  \cdots  I(1)_{\pi(n)}.
\end{equation}
In Eq.~\eqref{Apeq:Initial_state_expansion}, the first summation fixes the number of off-diagonal factors ($\sigma_\pm$) in a string. The second summation accounts  for all possible unique permutations for a $n$-qubit string with $r$ non-diagonal terms. The third summation fixes the choice of the non-diagonal terms, as each non-diagonal term can be either $\sigma_+$ or $\sigma_-$.    

Thus, the initial state is represented as an exponentially large sum of strings $s \in \mathcal{I}^{\otimes n}$. To each string in $\mathcal{I}^{\otimes n}$, we associate an overall coefficient $\beta$ and the vector $\vec{a} = (a_1, \dots,a_{n - r})$ denotes the parameters $a_i$ of the diagonal factors $I(a_i)$ in the string. Thus, each string is fully characterized by at most $n+1$ complex parameters. As we will see below, under the action of the noisy IQP circuit on any string $s$, only the parameters $\beta$ and $\vec{a}$ are updated.  We will use the notation $\beta^{(r)}$ and $\vec{a}^{(r)} = (a_1^{(r)}, \dots, a_m^{(r)})$ to denote the coefficients after $r$ layers of the noisy IQP circuit $\mathcal{C}$. For all strings in Eq. \eqref{Apeq:Initial_state_expansion}, the initial coefficients satisfy 
\begin{equation}
    \beta^{(0)} = \frac{1}{2^n}, \qquad \vec{a}^{(0)} = (1, \dots, 1). 
\end{equation}
\subsection{Evaluating single qubit gates and ``commuting" them through the circuit}
Recall that the diagonal gates in our gate set are drawn from the gate set $\mathcal{G}_2 = \{e^{i \theta Z}, C(\phi) \}$, where $C(\phi)$ is a controlled phase gate with arbitrary phase, 
\begin{equation}
    \label{Apeq:Controlled_Phase}
    C(\phi) = \ketbra{00}{00} + \ketbra{01}{01} + \ketbra{10}{10} + e^{i \phi} \ketbra{11}{11}.
\end{equation}
Single-qubit diagonal gates act trivially on diagonal operators and only modify the phase of off-diagonal components. Indeed, 
\begin{equation}
     e^{i \theta Z} \lf \sigma_{\pm }\ri  e^{-i \theta Z} = e^{\mp i 2 \theta} \sigma_{\pm}, \qquad  e^{i \theta Z} \lf I(a) \ri  e^{-i \theta Z} = I(a). 
 \end{equation}
 Thus, for any single-qubit diagonal gate on qubit $l$ acting on a string $s \in \mathcal{I}^{\otimes n}$,we have
 \begin{equation}
      e^{i \theta Z_l} \lf \beta s \ri  e^{-i \theta Z_l} = \beta\exp \Big[ i 2 \theta f(s) \Big] s, 
 \end{equation}
 where, 
 \begin{equation}
     f(s) = \left\{ \begin{matrix}
          1,~&~\text{if}~s_l =\sigma_-  , \\
          -1,~&~\text{if}~s_l =\sigma_+  , \\
          0,~&~~~\text{if}~s_l = I(a).
     \end{matrix} \right.
 \end{equation}
Hence, each single-qubit diagonal gate updates only the global coefficient $\beta$, which acquires a phase. Thus, the action of single qubit-gates are very efficient to keep track of with only one operation per gate required to update $\beta$. 

Since single-qubit $Z$-rotations commute with controlled-phase gates, their relative order is irrelevant. We now show that they may also be commuted through amplitude damping channels, allowing all single-qubit diagonal gates to be pushed to the end of the noisy circuit. This simplifies both the analysis and the implementation of the algorithm. 
 
 For any  channel $\mathcal{E}$ with Kraus operators $\{K_m | m = 1, \dots, M\}$ and $U$ a unitary, define the channel $\mathcal{F}$ with the Kraus operators $\{F_m = U^\dagger K_m U| m = 1, \dots, M\}$. Then, 
\begin{equation}
\label{Apeq:commute_trick}
    \mathcal{E} \circ \mathcal{U} (\rho) = \sum_{m=1}^M  K_m U \rho U^\dagger K_m^\dagger  = \sum_{m = 1}^M U U^\dagger K_m  U \rho U^\dagger  K_m^\dagger U U^\dagger  = \sum_{m = 1}^M U F_m \rho F_m^\dagger  U^\dagger = \mathcal{U} \circ \mathcal{F}  (\rho).
\end{equation}
 Eq. \eqref{Apeq:commute_trick} implies that we can replace the channel $\mathcal{E}$ that acts after $U$ by the channel $\mathcal{F}$ that acts before $U$. For amplitude damping, the Kraus operators are given by  
\begin{equation}
    K_ 0 = \begin{bmatrix}
        1 & 0 \\ 0 & \sqrt{1 - p} 
    \end{bmatrix} , \quad
    K_1 = \begin{bmatrix}
        0 & \sqrt{p} \\ 0 & 0 
    \end{bmatrix} = \sqrt{p} \sigma_-.
\end{equation}
 For a single-qubit diagonal gate,  
\begin{equation}
    R_z(\theta) = \exp \lf i \theta Z \ri = \cos(\theta) \mathbb{I} + i \sin(\theta) Z,
\end{equation}
the commuted channel has the Kraus operators, 
\begin{equation}
     \tilde{K}_0 = R_z^\dagger(\theta) K_0 R_z(\theta) = K_0, \quad \tilde{K}_1 = R_z^\dagger(\theta) K_1 R_z(\theta) = e^{-2i \theta}K_1.
\end{equation}
Note that, 
\begin{equation}
    \begin{bmatrix}
        \tilde{K}_0\\\tilde{K}_1
    \end{bmatrix} = \begin{bmatrix}
        1 & 0 \\ 0 & e^{-i2\theta} 
    \end{bmatrix} \begin{bmatrix}
            K_0 \\ K_1
        \end{bmatrix} .
\end{equation}
Since $\{\tilde{K}_0,\tilde{K}_1\}$ is unitarily equivalent to $\{K_0, K_1\}$, it defines the same amplitude-damping channel. Hence,  single-qubit $Z$ rotations can be freely interchanged with the amplitude-damping channel. 

\subsection{Evaluating the controlled phase gates}
Consider the circuit $\mathcal{C}$ written as
\begin{equation}
    \mathcal{C} = \mathcal{E} D_d  \cdots ~D_2 \mathcal{E} D_1. 
\end{equation}
Using the commuting trick we have established in the above subsection, we can re-write the circuit as follows, 
\begin{equation}
    \mathcal{C} = \textcolor{blue}{\mathcal{C}_2} \textcolor{orange}{\mathcal{C}_1}  = \textcolor{blue}{R_d ~  \cdots ~R_1~ } \textcolor{orange}{\mathcal{E} ~C_d   \cdots ~ \mathcal{E} ~C_1}, 
\end{equation}
where $C_r$ is the set of two-qubit controlled-phase gates acting on layer $r$ and $R_r$ is the set of single-qubit gates acting on layer $r$.  From the preceding discussion, the circuit component $\mathcal{C}_2$ can be evaluated efficiently for all strings in $\mathcal{I}^{\otimes n}$. We  therefore focus on evaluating $\mathcal{C}_1$. Using the expression of the controlled-phase gate in the computational basis (Eq.~\eqref{Apeq:Controlled_Phase}),  the following identities hold, 
\begin{align}
    \label{Apeq:identities_line_a}
    &C(\theta) I(a) \otimes  I(b) C^\dagger (\theta) = I(a) \otimes I(b), \\
    \label{Apeq:identities_line_b}
    &C(\theta) \sigma_{\pm} \otimes \sigma_{\pm} C^\dagger (\theta) =  e^{i (\theta/2)  \lf n_+  - n_-\ri } \sigma_{\pm} \otimes \sigma_{\pm}, \qquad  n_\pm = \# \sigma_{\pm}, \\ 
\label{Apeq:identities_line_c}
    &C(\theta) \sigma_{\pm} \otimes I(a)  C^\dagger (\theta)  =  \sigma_\pm \otimes I\lf ae^{ \pm  i \theta}\ri.
\end{align}
Eq.~\eqref{Apeq:identities_line_a} follows from the fact that $I(a)$ is diagonal. Eq.~\eqref{Apeq:identities_line_b} is obtained by expanding $\sigma_\pm$ in the computational basis and noting that $C(\theta)$ applies a non-trivial phase only on $\ket{11}$.  Eq.~\eqref{Apeq:identities_line_c}  follows directly from 
\begin{equation}
    \hspace{-0.3cm} C(\theta) \sigma_{+} \otimes I(a)  C^\dagger (\theta) = C(\theta) \lf \ketbra{10}{00} + a\ketbra{11}{01} \ri C^\dagger (\theta) = \ketbra{10}{00} + a e^{i \theta}\ketbra{11}{01} = \sigma_+ \otimes I\lf ae^{i \theta}\ri. 
\end{equation}
Similarly, the action of local amplitude damping channel on the elements of the frame is  
\begin{align}
    \label{Apeq:identities_line_d}
    & \mathcal{E}_p \lf \sigma_{\pm} \ri = \sqrt{1 - p} \sigma_{\pm}, \\
    \label{Apeq:identities_line_e}
    & \mathcal{E}_p \lf I(a) \ri = \mathcal{E}_p \lf \ketbra{0}{0} + a \ketbra{1}{1} \ri = (1 + a p) I\lf \frac{a(1 - p)}{1 + a p} \ri.
\end{align}
Thus, as noted in the main text, both the controlled-phase unitary and the amplitude-damping channel maps elements of   $\mathcal{I}^{\otimes n}$ to themselves. Moreover, they preserve the structural form of each string: qubits with diagonal components remain diagonal, and those with off-diagonal components remain off-diagonal. Consequently, the circuit acts solely by updating the coefficients  $\beta$ and $\vec{a}$. The task thus reduces to finding an efficient classical algorithm for updating these coefficients.

Let $s \in \mathcal{I}^{\otimes n}$ be an arbitrary $n$-qubit string with an associated coefficient $\beta$. Define the following sets associated with the string $s$, , 
\begin{equation}
    S_{+}^{(s)} :=\left \{ m \in \{1, \dots , n\} | s^{(m)} = \sigma_+ \right\}, \quad S_{-}^{(s)} :=\left \{ m \in \{1, \dots , n\} | s^{(m)} = \sigma_- \right\}, \quad S_{0}^{(s)} := \overline{S_{+}^{(s)} \cup  S_{-}^{(s)}}.
\end{equation}
These sets denote the indices of qubits with a $\sigma_+, \sigma_-$ or a diagonal component in string $s$, respectively. Given these sets, the operator can be explicitly written as
\begin{equation}
    s = \lf \bigotimes_{m \in S_+^{(s)}}  \sigma_{+,m}\ri \lf \bigotimes_{n \in S_-^{(s)}}  \sigma_{-,n}\ri \lf \bigotimes_{t \in S_0^{(s)}}  I_{t}(a_t)\ri, 
\end{equation}
where the notation $o_{n}$ denotes a single qubit operator $o \in \{\sigma_+, \sigma_-, I(a)\}$ acting on the $n^{th}$ qubit. Our algorithm evolves the strings layer by layer. At layer $r$, let $S^{(r)}_u$ denote the set of indices, where a non-trivial control unitary acts on  that layer, i.e, 
\begin{equation}
    S^{(r)}_u = \Big\{(x,y) \Big| 0 \leq x < y \leq n, C(\theta_{xy}) \text{ acts between qubits } (x,y)  \text{ on layer }r \Big\}.
\end{equation}
 Then the unitary gate layer can be written as, 
\begin{equation}
    C_r  = \prod_{(x,y) \in S^{(r)}_u} C \lf \theta_{xy} \ri.
\end{equation}
From Eqs.~\eqref{Apeq:identities_line_a}, \eqref{Apeq:identities_line_b}, and \eqref{Apeq:identities_line_c}, gates with non-trivial effects are the control unitaries acting on two $\sigma_+$'s , or two  $\sigma_-$'s, or a single $\sigma_\pm$ and a diagonal term. We thus concern ourselves with the following subsets of $S_u^{(r)}$, 
\begin{align}
    \label{Apeq:subsets_1}
    S_u^{(r)++} = S_u^{(r)} \cap \lf S_+^{(s)} \times S_+^{(s)}\ri&, \quad S_u^{(r)--} = S_u^{(r)} \cap \lf S_-^{(s)} \times S_-^{(s)} \ri, \\
    \label{Apeq:subsets_2}
    \quad S_u^{(r) + 0} = S_u^{(r)} \cap \lf S_+^{(s)} \times S_0^{(s)} + S_0^{(s)} \times S_+^{(s)} \ri&, \quad S_u^{(r) - 0} = S_u^{(r)} \cap \lf S_-^{(s)} \times S_0^{(s)} + S_0^{(s)} \times S_-^{(s)} \ri. 
\end{align}
Then, using the relations in Eqs.~\eqref{Apeq:identities_line_a}, \eqref{Apeq:identities_line_b} and \eqref{Apeq:identities_line_c}, we get  
\begin{align}
    C_r \beta s C_r^\dagger = \beta \exp \lf i\sum_{  (x,y) \in S_u^{(r)++}} \theta_{xy} - i\sum_{  (x,y) \in S_u^{(r)--}} \theta_{xy} \ri& \lf \bigotimes_{m \in S_0^{(s)}}  \sigma_+^{(m)}\ri \lf \bigotimes_{n \in S_0^{(s)}}  \sigma_-^{(n)}\ri \notag\\ 
    &\times \lf \bigotimes_{t \in S_0^{(s)}}  I^{(t)}\lf a_t \exp \lf i \sum_{  (x,t) \in S_u^{(r)+0}} \theta_{xt} - i \sum_{  (x,t) \in S_u^{(r)-0} } \theta_{xt} \ri  \ri\ri. 
\end{align}
And after applying the layer of noise, we have 
\begin{align}
    \label{eq:Final_expression}
    \mathcal{E}_p^{\otimes n} \lf C_r \beta s C_r^\dagger \ri = \beta (1 - p)^{\frac{\left|S_+^{(O)}\right| + \left|S_-^{(O)}\right|}{2}}&\lf \prod_{t \in S_0^{(s)}}  (1 + \tilde{a}_t p) \ri \exp \lf i\sum_{  (x,y) \in S_u^{(r)++}} \theta_{xy} - i\sum_{  (x,y) \in S_u^{(r)--}} \theta_{xy} \ri \notag\\ 
    &\times \lf \bigotimes_{m \in S_0^{(s)}}  \sigma_+^{(m)}\ri \lf \bigotimes_{n \in S_0^{(s)}}  \sigma_-^{(n)}\ri  \lf \bigotimes_{t \in S_0^{(s)}}  I^{(t)}\lf\frac{ \tilde{a}_t (1 - p)  }{1 + \tilde{a}_t p}\ri\ri. 
\end{align}
where
\begin{equation}
 \label{Apeq:a_tilde}
    \tilde{a}_t := a_t \exp \lf i \sum_{  (x,t) \in S_u^{(r)+0} } \theta_{xt} - i \sum_{  (x,t) \in S_u^{(r)+0} } \theta_{xt} \ri .
\end{equation}

\subsection{Scaling analysis}
Summarizing the above arguments, after a layer of controlled-phase gates $C^{(r)}$ and amplitude-damping noise, the coefficients $\beta$ and $\vec{a}$ change as 
\begin{equation}
    \hspace{-0.5cm}\beta \to \beta (1 - p)^{\left(\left|S_+^{(O)}\right| + \left|S_-^{(O)}\right|\right)/2}\left[ \prod_{t \in S_0^{(s)}}  \left (1 + \tilde{a}_t p  \right) \right ] \exp \lf i\sum_{  (x,y) \in S_u^{(r)++}} \theta_{xy} - i\sum_{  (x,y) \in S_u^{(r)--}} \theta_{xy} \ri, \quad     a_t \to \frac{ \tilde{a}_t (1 - p) }{1 + p \tilde{a}_t }, 
\end{equation}
where the symbol $\tilde{a}_t$ is defined in Eq.~\eqref{Apeq:a_tilde}.  These rules allow any string in $\mathcal{I}^{\otimes n}$ to be propagated through $\mathcal{C}_1$ in polynomial time, as shown below. The action of $\mathcal{C}_2$ can likewise be evaluated in polynomial time, since it only contains single-qubit unitaries. In Algorithm~\ref{alg:Evolution_of_a_layer} we provide an explicit procedure for updating the coefficients $\beta$ and $\vec{a}$ for an arbitrary string $s$ as it evolves under the circuit $\mathcal{C}$. We now analyze the scaling of Algorithm~\ref{alg:Evolution_of_a_layer}.  

We begin by remarking that through a relabeling of qubit indices, every operator string $s$ can be brought to the canonical form 
\begin{equation}
    \label{ApEq:Nice_form}
    s = \lf \bigotimes_{i = 1}^{k_+} \sigma_+ \ri \lf \bigotimes_{i = k_+ + 1}^{k_-} \sigma_- \ri \lf \bigotimes_{t = k_- + 1}^{n} I(a_t) \ri.
\end{equation}
For $n$-qubits, it takes a worst-case time scaling of $O(n)$ to calculate the required permutation map to bring any string $s$ to the form in Eq.~\eqref{ApEq:Nice_form}. In practice, we permute each string into the form in Eq.~\eqref{ApEq:Nice_form}, propagate it through the circuit, and then permute it back at the end. Thus, the qubit indices in each layer of unitary gate $S_u^{(r)}$ must also be permuted accordingly. Since each layer contains at most $n$ gates, this re-labeling can be done in $O(nd)$ for the entire circuit $\mathcal{C}$. These permutation steps are performed outside the core update routine. Hence, without loss of generality, we may assume that all input strings are already in the canonical form of Eq.~\eqref{ApEq:Nice_form} and that the gate indices have been permuted consistently. This can also be done without permuting by keeping track of a list that contains ${+,-,0}$ for each qubit index specifying the type of element at the index. Here you trade for time-complexity for space complexity. 

\begin{figure}[htbp!]
\centering
  \hrule 
    \vspace{-5pt}
   \refstepcounter{algocount}          
  \renewcommand{\figurename}{Algorithm} 
  \renewcommand{\thefigure}{\arabic{algocount}} 
  \renewcommand{\figurename}{Algorithm}
  \captionsetup{labelfont=bf, justification=centering, singlelinecheck=false}
  \caption{Algorithm for calculating the evolution of operator $s$ through circuit $\mathcal{C}$}
  \vspace{-1pt}
  \hrule
\vspace{5pt}
\begin{algorithmic}[1]
\STATE \textbf{Start}
\STATE Input string $s$ as the sets $S_+^{(s)}, S_-^{(s)}, S_0^{(s)}$.
\STATE Input overall coefficient $\beta = 1/2^n$ and vector $\vec{a} = (1,\cdots,1)$.
\STATE Input the total number of layers $d$. 
\STATE Input the controlled unitary operations on layer $r$, a set of pairs of indices $S^{(r)}_u$, and a list $\Theta^{(r)}$ containing the phases.
\STATE Input the single qubit unitary operations on layer $r$ a set of indices $R^{(r)}_u$, and the list $\Xi^{(r)}$ containing the phases.
\STATE Permute the qubit indices of the string $s$ and the gate layers $S^{(r)}_u$ to bring $s$ into the canonical form of Eq.~\eqref{ApEq:Nice_form}
\STATE Input noise strength $p$. 
\STATE \COMMENT {\textit{Evaluating Circuit} $\mathcal{C}_1$}
\FOR{$r$ \textbf{in} $1, \dots, d$}
\STATE \COMMENT {\textit{Evaluating the unitary layer $C_r$}}
\FOR{$(x,y)$ \textbf{in} $S_u^{(r)}$}
    \IF{$(x,y) \in S_+^{(s)} \times S_{+}^{(s)}$}
        \STATE $\beta \gets \beta \exp \lf i  \Theta^{(r)}_{xy}\ri$ 
        \ELSIF{$(x,y) \in S_-^{(s)} \times S_{-}^{(s)}$}
        \STATE $\beta \gets \beta \exp\lf-i \Theta^{(r)}_{xy}\ri$
    \ELSIF{$(x,y) \in S_{+}^{(s)} \times S_{0}^{(s)} \cup S_0^{(s)} \times S_{+}^{(s)} $}
        \IF{$x \in S_0^{(s)}$} 
            \STATE $x' \gets x$
        \ELSE 
            \STATE $x' \gets y$
        \ENDIF
        \STATE $a_{x'} \gets a_{x'} \exp \left( i \Theta^{(r)}_{xy} \right)$
    \ELSIF{$(x,y) \in S_{-}^{(s)} \times S_{0}^{(s)} \cup S_0^{(s)} \times S_{-}^{(s)} $}
        \IF{$x \in S_0^{(s)}$} 
            \STATE $x'\gets x$
        \ELSE 
            \STATE $x' \gets y$
        \ENDIF
        \STATE $a_{x'} \gets a_{x'} \exp \left(- i \Theta^{(r)}_{xy} \right)$
    \ENDIF
\ENDFOR
\STATE \COMMENT{\textit{Evaluating the layer of noise}}
\STATE $\lambda = 1$
\FOR{$a_t$ \textbf{in} $\vec{a}$ }
    \STATE $a_t \gets a_t(1 - p)/(1 + p a_t)$
    \STATE $\lambda  \to \lambda (1 + p a_t) $
\ENDFOR
\STATE $\beta \gets \beta\lambda(1 - p)^{\frac{1}{2} \left( |S_+^{(s)}| + |S_{-}^{(s)}| \right)} $
\ENDFOR

\COMMENT {\textit{Evaluating Circuit} $\mathcal{C}_2$ }
\FOR{$r$ \textbf{in} $1, \dots, d$}
\FOR{$x$ \textbf{in} $R_u^{(r)}$}
    \IF{$x \in S_+^{(s)} $}
        \STATE $\beta \gets \beta \exp \lf i  \Xi^{(r)}_{x}\ri$ 
    \ELSIF{$x \in S_-^{(s)}$}
        \STATE $\beta \gets \beta \exp\lf-i \Xi^{(r)}_{x}\ri$
    \ENDIF
\ENDFOR
\ENDFOR
\end{algorithmic}
\vspace{5pt}
  \hrule 
\label{alg:Evolution_of_a_layer}
\end{figure}

For a given string $s$, the core update routine of the algorithm proceeds by first evaluating $\mathcal{C}_1$ layer by layer, evaluating each gate serially. Because $s$ has been permuted to the canonical form of Eq.~\eqref{ApEq:Nice_form}, the cost of deciding the type of update scales as $O(1)$. This is because one only needs to compare the qubit index $x$ against two numbers $k_+$ and $k_-$  appearing in Eq.~\eqref{ApEq:Nice_form} to identify the operator at qubit $r$. 

The costliest step in the algorithm is updating the coefficients. Since the absolute values of the coefficients $\beta$ that we track is lower bounded by $O(2^{-n})$,  we require $O(n)$-bits of precision to store them. For a very weak upper bound, we require all our coefficients and phases to be represented with $O(n)$ bits of precision. The noise rate $p$ is constant and hence only requires a constant number of bits. Since cost of multiplying two numbers with $O(n)$ bits of precision is $O(n^2)$, the cost of updating the coefficients after each gate is $O(n^2)$. Since there are at most $n/2$ gates in a layer, the cost of simulating a single unitary layer is bounded by $O(n^3)$.  Using a similar reasoning, the cost of simulating the noise layer can also be bounded by $O(n^3)$. There are $d$ layers, so the combined cost of simulating $\mathcal{C}_1$ is bounded by  $O(n^3 d)$. A similar reasoning can be applied to circuit $\mathcal{C}_2$ as it updates $\beta$ with a multiplicative phase. Combining the cost permuting the input operator and gates, evaluation of circuit components $\mathcal{C}_1$ and $\mathcal{C}_2$,  the worst-case runtime of algorithm~\ref{alg:Evolution_of_a_layer} is $O(n^3d)$ for any given string $s$.

\subsection{Places for potential speedups}
Though we report that the worst-case scaling of the algorithm can be bounded by $O(n^3d)$, this bound is weak and can be strengthened in practice.  The coefficients $\vec{a}$  are initialized as $(1,\dots, 1)$ and decrease monotonically. Since we retain  HW strings with at most a $O(1)$ hamming weight, and since $\vec{a}$ coefficients only contribute to HW strings with $\ketbra{1}{1}$ terms, in practice, we believe it sufficient to store these coefficients to a constant factor precision $O(1)$. Similarly, for practical IQP circuits we would only implement phases that require $O(1)$ bits of precision. Hence, we only need to store $\beta$ with $O(n)$ precision. This immediately brings the cost of updating each coefficient down to $O(n)$ (multiplying a $O(n)$ bit-precision number with $O(1)$ bit-precision number). In fact, we can do better by working with  $\log(\beta)$ instead of $\beta$, noting that all update operations on $\beta$ are multiplication. This multiplication turns into addition and the number of bits of precision required to store this is only $O(\log(n))$. Thus for each gate, the cost of update can be brought down to $O(\log(n))$. Since there are at most $n/2$ gates in a layer, the cost of simulating a single unitary layer scales as $O(n \log n)$.  Since all coefficients that are updated in the noise layer use only $O(1)$ precision, each noise layer can be simulated in $O(n)$ time. There are $d$ layers, so the combined cost of simulating $\mathcal{C}_1$ scales as $O(n d\log n)$. The circuit $\mathcal{C}_2$ only updates $\beta$ with a multiplicative phase. There are atmost $nd$ single qubit gates in $\mathcal{C}_1$. Thus, the total phase to be multiplied can be calculated in $O(nd \log(n))$ time, where the additional $\log(n)$ factor comes from adding the total phase to $\log(\beta)$. Combining the permutation of gates and string, evaluation of circuit components $\mathcal{C}_1$ and $\mathcal{C}_2$,  the worst-case runtime of algorithm~\ref{alg:Evolution_of_a_layer} is now $O(n d\log n)$ for any string $s$.

\section{Bounding the trace of \texorpdfstring{$\sigma$}{sigma}}
\label{ap:Bounding_trace}
The trace of $\sigma$ is bounded fairly easily, by noting that,the IQP circuit gates leave the diagonal basis elements invariant.  
\begin{align}
   \label{ApEq:Trace_bound_line1}
    \operatorname{Tr} (\sigma) &= \frac{1}{2^n}\sum_{\substack{a \in \{0,1\}^n, \\ 2|a| \leq k}} \bra{a} \rho \ket{a},  \\
    \label{ApEq:Trace_bound_line2}
    &= 1 - \frac{1}{2^n}\sum_{\substack{a \in \{0,1\}^n, \\ 2|a| > k}} \bra{a}  \rho \ket{a}, \\
    \label{ApEq:Trace_bound_line3}
    &= 1 - \frac{1}{2^n}\sum_{r = (k+1)/2}^{n} {\binom{n}{r}} (2 - (1 - p)^d)^{n - r} (1 - p)^{dr}, \\
    \label{ApEq:Trace_bound_line4}
    &\geq  1 - \frac{\lf 2 - (1 - p)^d \ri^{n - (k + 1)/2}}{2^n} e^{n H\lf\frac{k+1}{2n} \ri} (1 - p)^{\frac{d(k+1)}{2}}. 
\end{align}
In the above chain of equations,  in line \eqref{ApEq:Trace_bound_line1}, we have used the fact that by definition, we have $\bra{a}\sigma\ket{a} = \bra{a}\rho\ket{a}$ is $2|a| \leq k$ and $\bra{a}\sigma\ket{a} = 0$ otherwise.  Line \eqref{ApEq:Trace_bound_line2}, is just inverting the sum and using the fact that the trace of $\rho$ must be one. In line \eqref{ApEq:Trace_bound_line3}, we have explicitly evaluated the trace, which is done below. Finally in line \eqref{ApEq:Trace_bound_line4}, we have used the Chernoff bound. This easy to see by noting that the second part in the Eq.~\eqref{ApEq:Trace_bound_line4} is similar to Eq.~\eqref{Eq:trunc_line7}. The only part that might be nontrivial is the explicit evaluation in \eqref{ApEq:Trace_bound_line3}. To see that, first note that for a single qubit, 
\begin{equation}
    \label{ApEq:expression_h5}
    \mathcal{E}_p^d (\ketbra{+}{+} ) = \frac{1}{2}\begin{bmatrix}
       2 - (1 - p)^d & (1 -p)^{d/2} \\ (1 - p)^{d/2} & (1 - p)^d. 
    \end{bmatrix}
\end{equation}
Thus, for any bit string $a$ with hamming weight $|a| \leq k/2$, we have, 
\begin{align}
\label{Apeq:trace_line1}
    \bra{a}\rho\ket{a} &=  \braket{a \left| \mathcal{C}\lf\ketbra{+}{+}^{\otimes n}  \ri  \right|a} \\
    \label{Apeq:trace_line2}
    &=\braket{a \left| \mathcal{E}_p^{d \otimes n}\lf\ketbra{+}{+}^{\otimes n}  \ri  \right|a} \\
    \label{Apeq:trace_line3}
    &= \prod_{i = 1}^{n} \braket{a_i \left| \mathcal{E}_p^{d }\lf\ketbra{+}{+}  \ri  \right|a_i} \\
    \label{Apeq:trace_line4}
    &= \lf 2 - (1 - p)^d \ri^{n - |a|} (1 - p)^{d |a|}
\end{align}
In the above chain of equations, in line \eqref{Apeq:trace_line1}, we have used the exact expression for $\rho$. In line \eqref{Apeq:trace_line2}, we have used the fact that diagonal gates leave the diagonal operators invariant and that the amplitude damping does not mix diagonal and non-diagonal terms. The notation $\mathcal{E}_p^{d \otimes n}$ is to denote the channel $\mathcal{E}_p^{d}$ applied on all $n$ qubits. In line \eqref{Apeq:trace_line3}, we have used that $\ket{a} = \bigotimes_{i  = 1}^n \ket{a_1}$. In line \eqref{Apeq:trace_line4}, we have used the expression in Eq.~\eqref{ApEq:expression_h5} and that $a$ contains $n - |a|$ bits of $0$ and $|a|$ bits of $|1|$. Finally, the combinatorial factor in Eq.~\eqref{ApEq:Trace_bound_line3} comes from the total number bit strings  with the same Hamming weight.

\section{Scaling of the cutoff Hamming weight \texorpdfstring{$k$}{k}} 
\label{ap:k_scaling}
In this section, we determine the maximum Hamming weight required in our truncation scheme for a fixed trace-distance error between 
the exact state $\rho$ and its truncated approximation $\sigma$. Requiring $k$  to remain $O(1)$ imposes a corresponding threshold on the circuit depth, which we also derive. We begin by recalling that the use of Chernoff bound in bounding the Hilbert-Schmidt error in Lemma~\ref{lem:Trunc_error} requires 
 \begin{equation}
     \label{Apeq:Chernoff}
     k+1 \geq  n(1 - p)^d .
 \end{equation}
 Taking the natural logarithm on both sides gives
 \begin{equation}
     d \geq \frac{\ln n  - \ln (k+1)}{\ln(1-p)^{-1}}, 
 \end{equation}
 from which  $k = O(1)$ implies that $d = \Omega(\log(n))$. We therefore parametrize the circuit depth as, 
 \begin{equation}
     \label{Apeq:Expression_d}
     d = \frac{\lambda}{\ln(1 -p)^{-1}} \ln(n)
 \end{equation}
Substituting Eq.~\eqref{Apeq:Expression_d} into Eq.~\eqref{Apeq:Chernoff} yields the weak threshold $\lambda \geq 1$ (noting that the maximum hamming weight $k \geq 0$).  We stress that $\lambda$ need not be a constant and can in fact scale with $n$. In fact, when $d$ scales as $O(\poly(n))$, so will $\lambda$.

Next, consider the rank of $\sigma$, which by definition, is less than the total number of non-zero coefficients in $\hwt{\sigma}$. For $(k+1)/2n \leq 1/2$, we obtain 
\begin{equation}
  \text{rank}(\sigma) \leq \sum_{r  =0}^{k} {\binom{2n}{r}} \leq e^{2n H\lf \frac{k+1}{2n} \ri}, 
\end{equation}
where the second inequality follows from a standard upper bound on partial sums of binomial coefficients (see Lemma 16.19 in Ref.~\cite{flum2006parameterized}). Hence, for any $ 0 < k + 1 \leq n $  
\begin{equation}
\label{ApEq:Bound_rank}
   \sqrt{\text{rank}(\sigma)} + 1 \leq e^{n H\lf \frac{k+1}{2n} \ri}  + 1 \leq 2e^{n H\lf \frac{k+1}{2n} \ri} . 
\end{equation}

Combining the bound on trace distance from Thm.~\ref{thm:Bound_on_TD}, Eq.~\eqref{ApEq:Bound_rank} and the bound on the Hilbert-Schmidt error from Lemma.~\ref{lem:Trunc_error}, we obtain     
 \begin{equation}
     \|\rho - \sigma\|_{TD} \leq \lf \sqrt{\text{rank}(\sigma)} + 1 \ri \|\rho - \sigma\|_{HS}\leq 2 (2 - (1 - p)^d)^{-(k+1)/2}e^{2nH\lf \frac{k+1}{2n}\ri} (1 - p)^{d(k+1)/2}.
 \end{equation}
Taking natural logarithm on both sides yields
\begin{equation}
    \label{eq:log_bs_ea}
    \ln \lf  \|\rho - \sigma\|_{TD} \ri \leq \ln 2 + \frac{(k+1)}{2}\ln \lf \lf 2 - (1 - p)^d \ri^{-1}\ri + 2nH\lf \frac{k+1}{2 n}\ri + \frac{(k+1)}{2} \ln \lf (1 - p)^{d} \ri .
\end{equation}
Using the inequality $- x \ln x \geq - (1 - x)\ln(1 - x)$, for $x \leq 0.5$, we obtain 
\begin{align}
    H\lf \frac{k+1}{2 n}\ri &= -\lf \frac{k+1}{2n}\ri \ln \lf \frac{k+1}{2n}\ri - \lf1 -  \frac{k+1}{2n}\ri \ln \lf 1 - \frac{k+1}{2n}\ri, \\
    &\leq -2 \lf \frac{k+1}{2n}\ri \ln \lf \frac{k+1}{2n}\ri, \label{Apeq:Bin_entropy_scaling_bound}
\end{align}
which, together with  $\lf 2 - (1 - p)^d \ri^{-1} \leq 1$, gives an upper bound on Eq.~\eqref{eq:log_bs_ea},
\begin{equation}
    \label{ApEq:TD_bound_in_ap}
    \ln \lf  \|\rho - \sigma\|_{TD} \ri \leq \ln(2) + \lf k+ 1 \ri \lf   2\ln \lf \frac{2n}{k+1}\ri + \frac{1}{2}\ln \lf 1 - p \ri^d \ri . 
\end{equation}
Let $1/2> \delta > 0$ be the tolerable fixed error on the trace distance. This can be achieved by choosing $k$ such that 
\begin{equation}
    \label{Apeq:first_bound_eq}
    \lf k+ 1 \ri \lf   2\ln \lf \frac{2n}{k+1}\ri + \frac{1}{2}\ln \lf 1 - p \ri^d \ri \leq \ln\lf \frac{\delta}{2}\ri, 
\end{equation}
 Re-writing Eq.~\eqref{Apeq:first_bound_eq} using the parametrisation for $d$ defined in Eq.~\eqref{Apeq:Expression_d}, we obtain 
\begin{equation}
    \label{ApEq:eps}
    \lf k+ 1 \ri \lf   2\ln \lf \frac{2}{k+1}\ri + \lf 2 - \frac{\lambda}{2} \ri \ln n \ri \leq \ln\lf \frac{\delta}{2}\ri. 
\end{equation}
Multiplying $-1$ on both sides 
\begin{equation}
   \label{Apeq:Almost_there}
    \lf k+ 1 \ri \lf \ln \frac{(k+1)^2 }{4}  + \lf \frac{\lambda}{2} - 2
    \ri  \ln \lf n \ri \ri \geq \ln \lf \frac{2}{\delta}\ri.
\end{equation}
The above equation can be solved exactly by setting LHS equal to RHS, using the Lambert W function. Define $\beta$ and $y$ as
\begin{equation}
    \beta := \frac{1}{2} \left[ \lf \frac{\lambda}{2} - 2 \ri \ln n - 2 \ln 2 \right], \quad y := e^{\beta } (k + 1). 
\end{equation}
Eq. \eqref{Apeq:Almost_there} is then recast as (for the case of equality)
\begin{equation}
    y \ln y = \frac{1}{2} \ln \lf \frac{2}{ \delta} \ri \exp{\beta}.
\end{equation}
Substituting $u = \ln y$, the above equation turns into the standard form of Lambert W function~\cite{corless1996lambert}. Since $y, \delta$ and $\beta$ are reals and $ \ln \lf ({2 \delta})^{-1} \ri \exp{\beta} \geq 0$, the solution is in the principle branch $W_0$. Thus, we have 
\begin{equation}
    y = \frac{1}{2} \ln \lf \frac{2}{ \delta} \ri \frac{\exp{\beta}}{W_0 \left[ \frac{1}{2} \ln \lf \frac{2}{ \delta}\ri \lf \frac{n}{4}\ri^{\frac{1}{4} (\lambda - 4)} \right]}. 
\end{equation}
This gives us the expression for $k$, 
\begin{equation}
    k + 1 = \frac{1}{2} \ln \lf \frac{2}{ \delta} \ri \left[ W_0 \left(\frac{1}{2} \ln \lf \frac{2}{ \delta}\ri \lf \frac{n}{4}\ri^{\frac{1}{4} (\lambda - 4)} \right) \right]^{-1}
\end{equation}
The principal branch of the Lambert W function has the following properties~\cite{corless1996lambert}.
\begin{equation}
    W_0(0) = 0, \quad \lim_{x \to \infty} W_0(x) = \ln x.
\end{equation}
This gives us the threshold on $d$. If $\lambda < 4$, then $ n^{(1/4) (\lambda - 4)} \to 0$ and if $\lambda > 4$, $ n^{(1/4) (\lambda - 4)} \to \infty$, as $n \to \infty$. Thus, for $\lambda > 4$, 
\begin{equation}
     k + 1 = O \lf \frac{\log \lf {2/\delta} \ri}{ \log\lf n \ri }\ri. 
\end{equation}
The scaling can be put in more rigorous grounds. First note that, for all $k \geq 0$, we have 
\begin{equation}
    \ln \frac{(k+1)^2}{4} \geq  -\ln \lf 4\ri.
\end{equation}
Then, letting the depth parameter $\lambda$ satisfy
\begin{equation}
    \label{ApEq:Min_expression}
    \lambda \geq 4 + 2\frac{\ln 4}{\ln (n)}, 
\end{equation}
ensures that demanding the following equation,
\begin{equation}
     \label{ApEq:Min_exp}
     \lf k+ 1 \ri \lf \lf \frac{\lambda}{2} -  2 \ri  \ln \lf n \ri  - \ln 4\ri =  \ln \frac{2}{\delta}
\end{equation}
we can satisfy Eq.\eqref{Apeq:Almost_there} by choosing a $k+1$ such that   
\begin{equation}
     \label{ApEq:Final_scaling}
    k + 1  =   \frac{ \ln (2/\delta)}{ \lf \frac{\lambda}{2} -  2 \ri  \ln \lf n \ri  - \ln 4} = O \lf \frac{\log \lf {2/\delta} \ri}{ \log\lf n \ri }\ri.
\end{equation}
Note that in the aymptotic limit, $n \to \infty$, we recover the threshold from the Lambert-W function. Putting everything together, for every $\delta$ we choose a $k$ that satisfies Eq. \eqref{ApEq:Final_scaling}, which ensures that Eq.~\eqref{ApEq:Min_exp} is satisfied. Exponentiating Eq.~\eqref{ApEq:TD_bound_in_ap}, we see that, 
\begin{equation}
    \|\rho - \sigma\|_{TD} \leq \delta.
\end{equation}
Thus, the trace distance between $\sigma$ and $\rho$ can be bounded by any $1/2 > \delta > 0$, as long as the circuit depth $d$ satisfies
\begin{equation}
    d \geq \frac{ 4 \ln(n) + 2\ln(4)  ,}{\ln(1-p)^{-1}} := d_T 
\end{equation}
by choosing the maximum cutoff as 
\begin{equation}
    k+1 = \frac{ \ln (2/\delta)}{ \lf \frac{\lambda}{2} -  2 \ri  \ln \lf n \ri  - \ln 4} = O \lf \frac{\log \lf {2/\delta} \ri}{ \log\lf n \ri }\ri.
\end{equation}
This implies that the total number of HW operators that we need to track scales as $O(n^k)$ is $O(\poly\lf2 / \delta \ri)$. Furthermore, note that for $d > d_T$, we have 
\begin{equation}
    \label{ApEq_kd_scaling}
    (k+1) d = \frac{1}{\ln(1 - p)^{-1}}\frac{ \ln (2/\delta)}{  \frac{1}{2}  -  \lf \frac{2}{\lambda} + \frac{\ln 4}{\lambda \ln(n)} \ri} = O \lf \frac{\log \lf {2/\delta} \ri}{ \log\lf 1-p \ri^{-1} }\ri. 
\end{equation}

\section{Worst case run-time for estimating generating \texorpdfstring{$Q_\mathcal{C}$}{Qc}}
\label{ap:Full_algorithm_scaling}
In this section, we estimate the worst case run-time for estimating $Q_\mathcal{C}$. The algorithm proceeds in two steps: (i) evaluate the truncated approximant $\sigma$ , and (ii) sample from it using the procedure of Ref.~\cite{BMS17}. We focus here on the cost of step (i). We start by considering the initial state expanded in the operator frame
\begin{equation}
\label{Apeq:Initial_state_expansion_1}
    \ketbra{+}{+}^{\otimes n} = \frac{1}{2^n} \sum_{r = 0}^{n} \sum_{s_1, \dots, s_r = +, -} ~\sum_{\pi \in \mathbb{P}^r_n}  ~ \sigma_{s_1}^{(\pi(1))}  \cdots \sigma_{s_r}^{(\pi(r))} I(1)^{(\pi(r+1))}  \cdots  I(1)^{(\pi(n))}.
\end{equation}
From the scaling analysis in Sec.~\ref{Ap:Low_Hamming_Weight_Strings}, each string in Eq.~\eqref{Apeq:Initial_state_expansion_1} can be propagated through $\mathcal{C}$ in time $O(n^3d)$. The first crucial observation here is that every operator in the HW basis representation of the initial state is contained in a unique string in Eq.~\eqref{eq:Initial_state_frame_rep}.  For example, the string $I(1)^{(1)}\dots I(1)^{(n)}$ contains all the operators in the HW basis consisting of only diagonal terms. Similarly the string $\sigma_+^{(1)}I(1)^{(2)}\dots I(1)^{(n)}$ contains all HW operators with $\sigma_+$ on the first qubit and diagonal terms elsewhere. More generally, if $o_i$ denotes an HW operator with an index $i$, then there exists a unique string $s_t$ and $m$ such that, 
\begin{equation}
    \operatorname{Tr}  \lf o_i^\dagger s_t \ri \neq 0, \quad s_t \in S_m,  \forall i.
\end{equation}
Here, $S_m$ is the set of all strings in Eq.~\eqref{Apeq:Initial_state_expansion_1} that contains $m$ non-diagonal terms. Moreover, the support of operators in the HW basis on a given string $s$ is preserved under propagation of the string under $\mathcal{C}$.  For example, propagating the string $I(1)^{(1)}\dots I(1)^{(n)}$ with an initial coefficient $1/2^n$, we get a new string of the form $I(a_1)^{(1)}\dots I(a_n)^{(n)}$, with some new coefficient $\beta'$. Crucially though, this new string still has the support of all the diagonal operators in the HW basis and only those.  Thus, our \textbf{first observation} is that, propagating $s$ and its associated coefficient $\beta$ through $\mathcal{C}$, allows us to exactly recover the coefficient $\alpha_i$ associated with $o_i$ for all $o_i$ contained in $s_t$, as
\begin{equation}
    \alpha^{(r)}_i = \beta^{(r)} \operatorname{Tr} \lf o_i^\dagger s_t^{(r)} \ri, 
\end{equation}
 where $ \beta^{(r)}(s_t^{(r)})$ is the coefficient (operator) after $r$ layers of evolution under $\mathcal{C}$.

The second crucial observation is that all elements to the set $S_m$ have $m$ factors of $\sigma_{\pm}$ terms. Thus, 
\begin{equation}
    [i] \geq m, ~ \forall \text{ operators } o \text{ such that } \operatorname{Tr}\lf o^\dagger s_t \ri \neq 0, s_t \in S_m.
\end{equation}
Again, for clarity let us work out some examples. The string the string $I(1)^{(1)}\dots I(1)^{(n)}$  contains no non-diagonal term. The operator with minimum hamming weight supported in this string is $\hwt{0\dots0}$. This has a hamming weight zero. Similarly, consider the string, $s = \sigma_+^{(1)}I(1)^{(2)}\dots I(1)^{(n)}$. It has one non-diagonal operator. Correspondingly, the operator with minimum weight that is supported in this string has the representation in computational basis $\ketbra{1 0 \dots 0}{0 0 \dots 0}$, with a hamming weight 1.  Thus our \textbf{second observation} is that, for any string $s$ in the expansion in Eq.~\eqref{Apeq:Initial_state_expansion}, the minimum hamming weight of the operators supported in $s$ is equal to the total number of non-diagonal elements in $s$.

Thus, by evolving all strings in the expansion in Eq.~\eqref{Apeq:Initial_state_expansion_1} with number of non-diagonal operators $m \leq k$,  we can access the coefficients of all HW operators with hamming weight lesser than or equal to $k$. From the bound on truncation error, to estimate $\sigma$ we only need propagate  $k = O \lf \frac{\log \lf {2/\delta} \ri}{ \log\lf n \ri }\ri$, which implies a the total number of strings is 
\begin{equation}
    \sum_{i = 0}^{k} 2^i {\binom{n}{i}} = O(n^k) = O(\poly\lf2 / \delta \ri). 
\end{equation}
In the above equation, the $2^i$ comes from the fact that there are two non-diagonal operators. In practice, we can get some improvements by noting that certain strings are hermitian conjugates of each other. Nonetheless, since each string can be evaluated in  $O(n^3d)$  time and since we have to evaluate only $O(\poly\lf2 / \delta \ri)$ strings, the worst case runtime scales as   $O(n^3d ~\poly\lf2 / \delta \ri)$. Finally we need to calculate the coefficients $\alpha_i$ from the propagated strings. Since there are at most $O(\poly(\delta^{-1}))$ strings, this only adds a constant overhead. Finally, given $\sigma$ we can generate the quasi-probability distribution $q$, which has at most $O(\poly(\delta^{-1}))$ Fourier coefficients. Using the arguments from Ref \cite{BMS17}, we note that this implies that the marginals of $q(x)$ can be computed efficiently, using Parseval's identity. Let $y \in \{0,1\}^k$ be an arbitrary $k$-bit string, then the marginal $S_y$ is given as in Ref.~\cite{BMS17},  
\begin{equation}
    S_y = \sum_{x:x_{1 \dots k } = y} q(x) = 2^{n-k} \sum_{s: s_{k+1,\dots,n} = 0^{n - k}} \tilde{q}(s)(-1)^{y\cdot s_1 \dots s_k}, 
\end{equation}
where $\tilde{q}(s)$ are the Fourier coefficients of $q(x)$. Since the Fourier spectrum is sparse, the marginals are exactly computable in $O(n \poly(\delta^{-1}))$ time. The factor $n$ comes from the worst-case cost of calculating $(-1)^{y\cdot s_1 \dots s_k}$. Finally, the algorithm to sample from the distribution $Q_{\cal{C}}$, reproduced from Ref.~\cite{BMS17}, is 
\begin{enumerate}
    \item Set $y$ to the empty string
    \item For $i = 1,\dots , n$:
    \begin{enumerate}
        \item If $S_{yz} < 0$ for some $z \in \{0, 1\}$, set $y \longleftarrow y\bar{z}$,  where $\bar{z} = 1 - z$.
        \item Otherwise: with probability $S_{y0}/S_y$, set $y \longleftarrow y0$; otherwise, set $y \longleftarrow y1$.
    \end{enumerate}
    \item Return $y$
\end{enumerate}
Thus, each sample from $Q_{\mathcal{C}}$ can be produced with a worst-case time scaling of $O(n^2 \poly(\delta^{-1}))$.

\section{Extending the proof of Theorem \ref{thm:main_theorem} to \texorpdfstring{$\ell$}{l}-local gate set}
\label{ap:Theorem_main_extension}
In this section, we generalize the proof Theorem~\ref{thm:main_theorem} to $l$-local gate sets. As a natural extension of Theorem~\ref{thm:main_theorem}, we relax the two-qubit gate set restriction imposed therein. In particular, while a single element of the operator frame $\mathcal{I}^{\otimes n}$ maps to a unique element under the $2$-local gates considered in the main text, higher-order gates such as CCZ can map a single frame element to multiple elements. We show, however, that for an arbitrary $\ell$-local controlled-phase gate, this increased branching results only in a constant-factor overhead in the number of terms that must be tracked. We are able to show this, as our truncation scheme depends only on the circuit depth $d$ and is agnostic to the locality of the diagonal gates. This is because we bound the absolute values of the coefficients in the HW basis, hence the phase applied by the diagonal gates become irrelevant. The high-level proof structure is as follows.  First, we show that the action of an $\ell$-local controlled-phase gate on a frame element can be expressed as a linear combination of at most $2^{\ell - 1} + 1$ frame elements. We then show that for each unitary layer the blow-up due to branching is at most $(2^{\ell - 1} + 1)^k$, where $k$ is the maximum Hamming weight enforced by the truncation. Thus, for $d$ layers we have an overhead of at most $(2^{\ell - 1} + 1)^{kd}$ terms per frame element tracked. We complete the proof by noting that from Eq.~\eqref{ApEq_kd_scaling}, we have $(2^{\ell - 1} + 1)^{kd} = O(1)$.

We begin by defining the $\ell$-local controlled-phase gate by its action on the computational basis states, 
\begin{equation}
    C^{(\ell)}(\phi) \lf \ket{a_1\dots a_l} \ri = \left\{ \begin{matrix} e^{i \phi } \ket{a_1\dots a_l} & \text{ if } (a_1\dots a_\ell) = (1, \dots, 1),  \\
    \ket{a_1\dots a_\ell} & \text{ otherwise.}
    \end{matrix} \right.
\end{equation}
Define the notation $[n]:=\{1, \dots, n\}$. Now consider the action of the unitary on the following elements of the operator frame $\mathcal{I}^{\ell}$, 
\begin{equation}
   C^{(\ell)}(\phi)\left(I(a_1)\dots I(a_\ell)\right)C^{^\dagger(\ell)} (\phi)= I(a_1)\dots I(a_\ell),
\end{equation}
and
\begin{equation}
   \label{ApEq:muti-qubit_blowup}
   C^{(\ell)}(\phi)\left(\sigma_{\pm}\dots\sigma_{\pm}I(a_1)\dots I(a_m)\right)C^{^\dagger(\ell)} (\phi)= \sigma_{\pm}\dots\sigma_{\pm}I(a_1)\dots I(a_m)+(e^{i\phi}-1)a_1\dots a_m\sigma_{\pm}\dots\sigma_{\pm}\ketbra{1}{1}^{\otimes m}
\end{equation}
In other words $C^{(\ell)}(\phi)$ acts non-trivially only if at least one factor of the frame element is non-diagonal. Note that we need to represent the expression in Eq.~\eqref{ApEq:muti-qubit_blowup} as a linear combination of frame elements. We prove the following lemma for that purpose.  
\begin{lemma}
    \label{Aplemma:conversion_to_boolean_functions}
     The operator $ (e^{i\phi}-1)a_1\dots a_m\sigma_{\pm}\dots\sigma_{\pm}\ketbra{1}{1}^{\otimes m}$ can be expressed as a linear combination of frame elements
     \begin{equation}
         \label{ApEq:Defn_equal_op}
         (e^{i\phi}-1)a_1\dots a_m\sigma_{\pm}\dots\sigma_{\pm}\ketbra{1}{1}^{\otimes m} = \sum_{r = 0}^R \beta_r \sigma_{\pm}\dots\sigma_{\pm} \bigotimes_{j = 1}^m I\lf b_j^{(r)} a_j\ri,    
     \end{equation}
     if and only if
     \begin{equation}
        \label{ApEq:Implication_equal_op}
         \sum_{r = 0}^R \beta_r \prod_{i =1}^m (1 + b_i^{(r)} a_i) = (e^{i \phi} - 1) a_1 \cdots a_m
     \end{equation}
\end{lemma}
\textbf{Proof:} Consider the following expansion of the diagonal frame element, 
    \begin{equation}
        \bigotimes_{j = 1}^m I\lf b_i^{(r)} a_j\ri  = \sum_{S \subseteq [m]} \prod_{i \in S} b_i^{(r)} a_i \bigotimes_{i \notin S}
        \ketbra{0}{0} \bigotimes_{i \in S} \ketbra{1}{1} .   
    \end{equation}
  Substituting the above equation in Eq.~\eqref{ApEq:Defn_equal_op} we get, 
  \begin{equation}
     (e^{i\phi}-1)a_1\dots a_m\sigma_{\pm}\dots\sigma_{\pm}\ketbra{1}{1}^{\otimes m} =  \sum_{S \subseteq [m]}  \left[ \lf \sum_{r = 0}^R \beta_r \prod_{i \in S} b_i^{(r)} a_i \ri\sigma_{\pm}\dots\sigma_{\pm}  \bigotimes_{i \notin S} \ketbra{0}{0} \bigotimes_{i \in S} \ketbra{1}{1}\right], 
  \end{equation}
 This implies that 
 \begin{equation}
         \sum_{r = 0}^R \beta_r \prod_{i \in S} b_i^{(r)} a_i = \left\{ \begin{matrix}
             (e^{i\phi}-1) a_1\dots a_m  & \text{ if } S = [m],  \\
             0 & \text{ otherwise.}
         \end{matrix} \right. := (e^{i\phi}-1) a_1\dots a_m  \delta_{S,[m]}, \quad \text{for any } S \subseteq [m]. 
  \end{equation}
 We complete the proof by noting, 
 \begin{equation}
     \sum_{r = 0}^R \beta_r \prod_{i =1}^m  (1 + b_i^{(r)} a_i)=  \sum_{S \subseteq [m]} \sum_{r = 0}^R \beta_r \prod_{i \in S} b_i^{(r)} a_i = \sum_{S \subseteq [m]} (e^{i\phi}-1) a_1\dots a_m  \delta_{S,[m]} = (e^{i\phi}-1) a_1\dots a_m.
 \end{equation}
 Note that the R.H.S of Eq.~\eqref{ApEq:Implication_equal_op} has the form of an unnormalised boolean delta function~\cite{odonnell2021analysisbooleanfunctions}. We will use this equivalence to prove the following proposition that provides an explicit construction for the decomposition in Eq.~\eqref{ApEq:Defn_equal_op} with $R = 2^m$. 
 \begin{prop} 
 \label{ApProp:delta_expansion}
 Using Fourier transform on boolean functions, we can show that 
   \begin{equation}
        (e^{i\phi}-1) a_1\dots a_m  =  \sum_{\mathbf{x} \in \{0,1\}^m} \frac{(e^{i\phi}-1) (-1)^{\sum_{i =1}^m x_i} }{2^m}  \prod_{i = 1}^m\lf  1 + (-1)^{x_i} a_i \ri .  
 \end{equation}
 \end{prop}
  \textbf{Proof:} Consider the parity functions $\chi_{S}$ (see Definition 1.2 in Ref.\cite{odonnell2021analysisbooleanfunctions}) for some subset $S\subseteq [n]$,
  \begin{equation}
      \chi_{S} (\mathbf{x}) = (-1)^{\sum_{i \in S} x_i}.
  \end{equation}
Expanding $ \prod_{i = 1}^m\lf  1 + (-1)^{s_i} a_i \ri$, we have 
\begin{equation}
      \label{ApEq:Expansion_2}
     \prod_{i = 1}^m\lf  1 + (-1)^{x_i} a_i \ri = \sum_{T \subseteq [m]}  \prod_{i \in T} (-1)^{x_i} a_i =  \sum_{T \subseteq [m]} (-1)^{\sum_{i \in T} x_i} \prod_{i \in T} a_i =  \sum_{T \subseteq [m]} \chi_T( \mathbf{x}) a_T, 
\end{equation}
where we have defined the notation $a_{T} := \prod_{i \in T} a_i $. Thus, 
\begin{align}
    \label{ApEq:Walsh_Had_line1}
    \frac{(e^{i\phi}-1)}{2^n} \sum_{\mathbf{x} \in \{0,1\}^m} (-1)^{\sum_{i =1}^m x_i} \prod_{i = 1}^m\lf  1 + (-1)^{x_i} a_i \ri &=(e^{i\phi}-1) \sum_{T \subseteq [m]}  \lf \frac{1}{2^m} \sum_{\mathbf{x} \in \{0,1\}^m} \chi_{[m]}( \mathbf{x}) \chi_T( \mathbf{x}) \ri  a_T, \\
     \label{ApEq:Walsh_Had_line2}
    &= (e^{i\phi}-1)\sum_{T \subseteq [m]} \delta_{T,[m]}  a_T, \\
     \label{ApEq:Walsh_Had_line3}
    &=(e^{i\phi}-1) a_1 \cdots a_m,
\end{align}
where in line \eqref{ApEq:Walsh_Had_line1} we have substituted the expanded form in Eq.~\eqref{ApEq:Expansion_2} and in line \eqref{ApEq:Walsh_Had_line2}, we have used the orthogonality of parity functions (see Theorem 1.5 in Ref.~\cite{odonnell2021analysisbooleanfunctions}). This completes the proof of the proposition.

Finally, note that using Prop.~\ref{ApProp:delta_expansion}, we can construct the linear combination in Eq.~\eqref{ApEq:Defn_equal_op} by setting
\begin{equation}
    \beta_r =  \frac{(e^{i\phi}-1) (-1)^{\sum_{i =1}^m x_i} }{2^m}  \text{ and } b_i^{(r)} = (-1)^{x_i}, 
\end{equation}
where $r$ is the decimal representation of the bit-string $\textbf{x}$. Thus, combining Lemma~\ref{Aplemma:conversion_to_boolean_functions} and Proposition~\ref{ApProp:delta_expansion} we have shown that  Eq.~\eqref{ApEq:muti-qubit_blowup} can be expressed by a linear combination of at most $2^m + 1$ frame elements. For an $\ell$-local controlled-phase gate, this implies that the maximum branching is $2^{\ell - 1} + 1$, which happens when the gate acts on an operator with a single non-diagonal factor. We note that this is a very loose bound. In fact, for $l = 2$, we have previously shown that there is no branching and strings are mapped one-to-one.  From the analysis in Sec.~\ref{ap:k_scaling} and the arguments presented in Sec. \ref{ap:Full_algorithm_scaling}, we have to only evolve strings with at most $k$ non-diagonal factors. Thus, the maximum number of branchings of a frame element (that we keep track of) under a single unitary layer is given by 
\begin{equation}
    \min[(2^{\ell - 1} + 1)^k, (2^{\ell - 1} + 1)^{\lfloor n/l\rfloor} ].
\end{equation}
This is because the maximum branching occurs when each $\ell$-local controlled-phase gate in a single layer acts on a single non-diagonal tensor factor. Since we need to keep only $k < O(1)$ for $d > d_T$, without loss of generality, $\min[(2^{\ell - 1} + 1)^k, (2^{\ell - 1} + 1)^{\lfloor n/l\rfloor} ] = (2^{\ell - 1} + 1)^k$. Using Eq.~\eqref{ApEq_kd_scaling}, for a constant noise $p$ and a constant error $\delta$,  the maximum number of terms that we have to keep track of, for a circuit of depth $d$ is given by
\begin{equation}
    (2^{\ell - 1} + 1)^{kd} = (2^{\ell - 1} + 1)^{O(\log(1/\delta))}.
\end{equation}
Hence, for a $\ell$-local gate set, we need to keep track of at most $(2^{\ell - 1} + 1)^{O(\log(1/\delta))} = O(\poly(1/\delta))$ terms per frame element as opposed to a single string for single and two-qubit gates. The rest of the analysis follows the arguments presented for two-qubit and single-qubit gates.  Therefore, the worst case runtime of our algorithm still only scales as $O\left(n^3 d \poly(1/\delta)\right)$.

\section{Numerical Demonstration of Algorithm for Two-Qubit Gate Set}
\label{ap:numerical}

In this section, we numerically demonstrate the validity of our analytic bounds for the $\mathcal{G}_2$ gate set. We simulate instances of random circuits that are composed of unitary layers drawn from $\mathcal{G}_2$, each followed by a layer of amplitude-damping channels acting on each qubit. The parameters in Fig.~\ref{fig:random_error_plot} are $n=10$, $d=10$, and $p=0.1$ with a probability of $1/2$ of choosing either a random controlled-phase two-qubit gate or two random single-qubit phase-gates for each pair of sites in each circuit layer. The analytically derived Hilbert-Schmidt error bound, given in Lemma~\ref{lem:Trunc_error}, is shown with the green stars and solid green curve. The blue curves in Fig.~\ref{fig:random_error_plot} correspond to the Hilbert-Schmidt error between the truncated approximations and the exact states at the end of the circuit, whereas the trace-distance error is shown in the red curves. The plotted points correspond to the mean values of $200$ random circuit instances and the error bars give the minimum and maximum observed error over all the instances. Furthermore, there are two different types of approximations shown in the figure. The dashed lines with circular points correspond to a truncated approximation that keeps truncates up to the associated Hamming weight. This is done by running Algorithm~\ref{alg:Evolution_of_a_layer} with a $\sigma_{\pm}$ weight of $k$ and then reproducing the state with a HW basis cutoff of $k$. The dotted lines with square points show an approximation that reproduces an approximation of the desired state $\rho$ with no HW cutoff, but reproduced only from the simulated frame elements up to a $\sigma_{\pm}$-weight of $k$. While the latter approximation is not efficient when mapped back to the HW basis (it contains exponentially many different HW strings), it nonetheless gives us a picture of the extent to which the low-weight frame elements capture the full state.

The first thing to note from Fig.~\ref{fig:random_error_plot} is that the Hilbert-Schmidt error bound holds for all of the numerically simulated circuit instances. Moreover, the bound is quite loose, and in practice, the algorithm performs better on average than the bound implies. The second interesting result is that the calculated trace distance also lies below the Hilbert-Schmidt error bound, which is why we do not bother plotting the much looser trace distance upper bound derived in the main text. While our arguments have a multiplicative overhead between the trace distance and HS error, proportional to the rank of the approximation $\sigma$, the trace distance and HS error are actually quite close. The argument in Theorem~\ref{thm:Bound_on_TD} is most tight when the state $\rho$ is near the maximally-mixed state. However, due to the nature of amplitude-damping noise, $\rho$ has exponentially little support outside of the low HW subspace. This exponentially-decaying envelope results in the trace-distance between $\rho$ and $\sigma$ being orders of magnitude better than predicted by Theorem~\ref{thm:Bound_on_TD}. The third observation is that the low $\sigma_\pm$-weight frame elements capture significantly more of the state than just the low HW elements. This makes sense because the frame elements, up to $\sigma_\pm$ weight of $k$, fully capture the HW strings up to weight $k$, but additionally also track an exponential number of higher HW strings.

\begin{figure}
    \centering
    \includegraphics[width=0.8\linewidth]{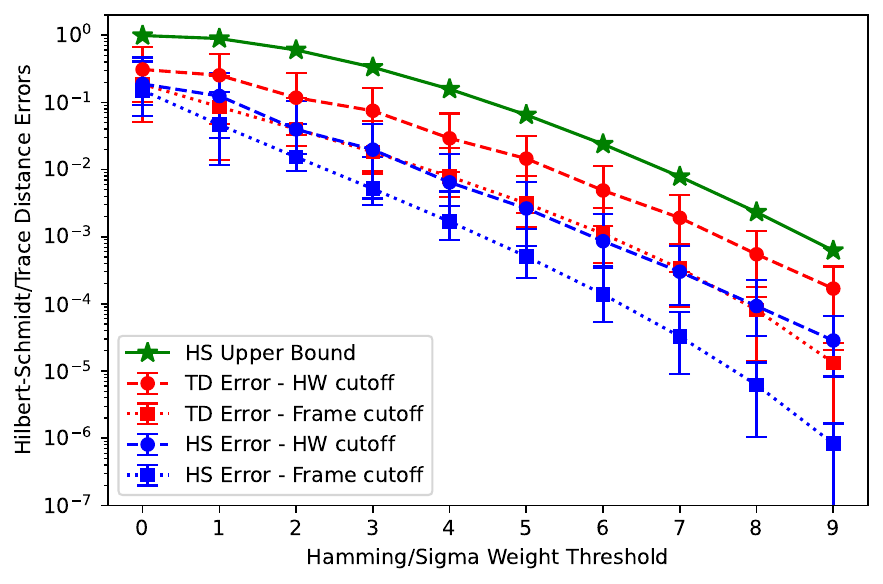}
    \caption{Hilbert-Schmidt error and trace-distance error between various approximations using our algorithm and the exact output state. The green stars show the Hilbert-Schmidt error upper bound that is given in Lemma ~\ref{lem:Trunc_error}. The blue points show the Hilbert-Schmidt error and the red points show the trace-distance error. Furthermore, the dashed lines with circular points, give the errors associated with the corresponding HW cutoff, whereas the dotted lines with square points give the errors associated with the corresponding frame-element weight cutoff. We used the following parameters: $n=10$, $d=10$, and $p=0.1$. Each point is the mean of $200$ random circuit instances drawn from $\mathcal{G}_2$ with a probability of $1/2$ of choosing either a random two-qubit controlled-phase gate or a single-qubit gate. The error bars represent the minimum and maximum errors observed.}
    \label{fig:random_error_plot}
\end{figure}

Another interesting consideration is that the worst-case circuit for our scheme appears to be the idle circuit, which is purely composed of amplitude-damping channels, without any unitaries. Shown in Fig.~\ref{fig:idle}, are the Hilbert-Schmidt errors of the random circuits (blue) alongside the Hilbert-Schmidt error of the idle circuit (black). For each HW truncation, we observe that the idle circuit error lies below the analytic upper bound, but above the maximum error instance of the random circuits. We expect this behavior since, as noted in Sec.~\ref{ap:bound_based_00_terms}, Proposition~\ref{prop:Truncation_error} is saturated by the idle circuit. A rigorous proof of this observation would enable a tighter bound on the Hilbert-Schmidt error. Further, if this observation holds even for error in trace-distance, then using the resulting tight bound, we might be able to reduce the depth threshold for classical simulability beyond the result reported in this Letter.

\begin{figure}
    \centering
    \includegraphics[width=0.6\linewidth]{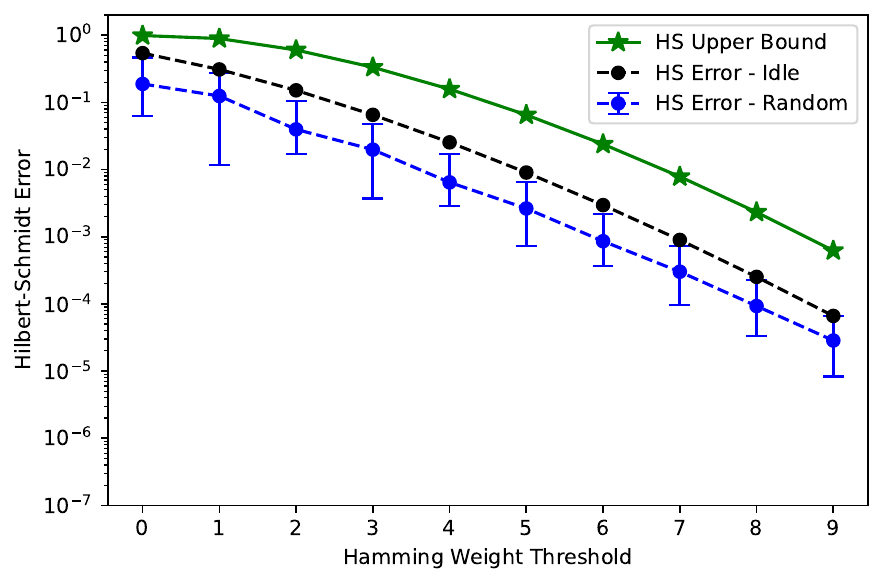}
    \caption{Hilbert-Schmidt error between the approximated state (with different Hamming Weight cutoffs) and the exact state for random circuits and the idle circuit. The green stars show the Hilbert-Schmidt error upper bound that is given in Lemma ~\ref{lem:Trunc_error}. The blue points show the mean Hilbert-Schmidt error for the random circuit instances, where the error bars give the minimum and maximum observed errors. Lastly, the black points show the Hilbert-Schmidt error for the idle circuit. We used the following parameters: $n=10$, $d=10$, and $p=0.1$. The blue points are the mean of $200$ random circuit instances drawn from $\mathcal{G}_2$ with a probability of $1/2$ of choosing either a random two-qubit controlled-phase gate or a single-qubit gate.}
    \label{fig:idle}
\end{figure}

\end{document}